\RequirePackage{amsmath}
\documentclass[letterpaper,11pt,english]{article}
\usepackage[letterpaper, portrait, margin=2.5cm]{geometry}

\usepackage{amssymb, amsfonts}
\usepackage{graphicx}
\usepackage{amsthm}
\usepackage{algorithm}
\usepackage[noend]{algpseudocode}
\usepackage[nolist]{acronym}
\usepackage[capitalize,nameinlink,noabbrev]{cleveref}
\usepackage{todonotes}
\usepackage{mathtools}
\usepackage{thm-restate}
\newacro{ses}[SES]{smallest enclosing sphere}

\usepackage{todonotes}

\newtheorem{theorem}{Theorem}
\newtheorem{lemma}[theorem]{Lemma}

\newtheorem{corollary}[theorem]{Corollary}
\newtheorem{definition}{Definition}

\def\Fsync/{$\mathcal{F}$\textsc{sync}}
\def\Gathering/{\textsc{Gathering}}
\def\chainForm/{\textsc{Chain-Formation}}
\def\patternForm/{\textsc{Pattern Formation}}
\def\Point/{\textsc{Point}}
\def\UniformCircle/{\textsc{Uniform Circle}}
\def\gtc/{\textsc{Go-To-The-Center}}
\def\gtcShort/{\textsc{GTC}}
\def\mobs/{\textsc{Move-on-Bisector}}
\def\gtcThreeD/{\textsc{3d-Go-To-The-Center}}
\def\gtcThreeDShort/{\textsc{3d-GTC}}
\def\gtcThreeDCont/{\textsc{Continuous-3d-Go-To-The-Center}}
\def\gtcThreeDContShort/{\textsc{Cont-3d-GTC}}
\def\tangentialNormal/{tangential-normal}
\def\moveOnAngleMinimizer/{\textsc{Move-on-Angle-Minimizer}}
\def\look/{\texttt{Look}}
\def\compute/{\texttt{Compute}}
\def\move/{\texttt{Move}}

\def\LCMlong/{\textsc{Look-Compute-Move}}
\def\LCM/{LCM}

\newcommand{\fsync}{\texorpdfstring{\textsc{$\mathcal{F}$sync}}{FSYNC}}

\newcommand{\ubg}[1]{\ensuremath{\mathrm{UBG}_{#1}}}
\newcommand{\ubgPlain}{\ensuremath{\mathrm{UBG}}}

\DeclareMathOperator{\ch}{CH}
\DeclareMathOperator{\pch}{PCH}
\DeclareMathOperator{\corn}{Corn}
\DeclareMathOperator{\adj}{Adj}
\DeclareMathOperator*{\argmin}{argmin}

\makeatletter
\def\@parfont{\bfseries}
\makeatother

\usepackage{authblk}

\title{Local Gathering of Mobile Robots in Three Dimensions \newline (Full Version)
\thanks{This paper is a full version of the respective paper presented at SIROCCO 2020.
}}
\author[]{Michael Braun}
\author[]{Jannik Castenow}
\author[]{Friedhelm Meyer auf der Heide}
\affil[]{Heinz Nixdorf Institute and Department of Computer Science\\
	Paderborn University, F\"urstenallee 11, 33102 Paderborn, Germany
}
\affil[]{ \{braunm, janniksu, fmadh\}@mail.upb.de}
\date{}

\begin{document}

\maketitle

\begin{abstract}
In this work, we initiate the research about the \Gathering/ problem for robots with limited viewing range in the three-dimensional Euclidean space.
In the \Gathering/ problem, a set of initially scattered robots is required to gather at the same position.
The robots' capabilities are very restricted -- they do not agree on any coordinate system or compass, have a limited viewing range, have no memory of the past and cannot communicate.

We study the problem in two different time models, in \fsync{}  (fully synchronized discrete rounds) and the continuous time model.
For \fsync{}, we introduce the \gtcThreeD/-strategy and prove a runtime of $\Theta\left(n^2\right)$ that matches the currently best runtime bound for the same model in the Euclidean plane [SPAA'11] .

Our main result is the generalization of contracting strategies (continuous time model) from [Algosensors'17] to the three-dimensional case.
In contracting strategies, every robot that is located on the global convex hull of all robots' positions moves with full speed towards the inside of the convex hull.
We prove a runtime bound of $\mathcal{O}\left(\Delta \cdot n^{3/2}\right)$  for \emph{any} three-dimensional contracting strategy, where $\Delta$ denotes the diameter of the initial configuration.
This comes up to a factor of $\sqrt{n}$ close to the lower bound of $\Omega \left(\Delta \cdot n\right)$ which is already true in two dimensions.

In general, it might be hard for robots with limited viewing range to decide whether they are located on the global convex hull and which movement maintains the connectivity of the swarm, rendering the design of \emph{concrete }contracting strategies a challenging task.
We prove that the continuous variant of \gtcThreeD/ is contracting and keeps the swarm connected.
Moreover, we give a simple design criterion for three-dimensional contracting strategies that maintains the connectivity of the swarm and introduce an exemplary strategy based on this criterion.
\end{abstract}


\section{Introduction}

We study a scenario where a distributed system of mobile entities (called \emph{robots}) is supposed to establish a certain formation, also denoted as a \emph{pattern}.
The robots are scattered in a $d$-dimensional Euclidean space (usually the Euclidean plane) and have to coordinate their movements in a distributed manner to reach the desired formation.
The robots' capabilities depend on the exact model and formation problem but are typically very restricted.
Usually, the robots do not agree on a common coordinate system or compass, cannot communicate with each other and have only limited sensing capabilities.
One extensively studied coordination problem is the \patternForm/ problem, dealing with questions such as: Which patterns are generally formable by a set of robots? Which capabilities do the robots need?
Given a specific pattern, for which initial configurations is this pattern formable?
Interestingly, it has been proven that there are only \emph{two} patterns that might be formable starting in an arbitrary input configuration.
These are the patterns \Point/ and \UniformCircle/.
Forming the pattern \Point/ is known under a more common name -- the \Gathering/ problem, which studies the task of gathering a set of robots on the same position.
Both of these problems have been extensively studied under several different assumptions, involving the viewing range (local or global), the synchronization (synchronous or asynchronous activation), the extent (robots can or cannot occupy the same position) or the opacity of robots, to name only a few.
However, most of these models have in common that the robots operate in the two-dimensional Euclidean plane.
A natural extension would be to consider the three-dimensional Euclidean space, where the robots have the ability to fly, such as drones, or to move underwater.
Existing results about robots in the three-dimensional Euclidean space are very scarce, rely on strong assumptions (such as axis agreement) and do not consider any runtime analyses of the proposed strategies.
Our work initiates the study of \Gathering/ of robots in three-dimensions, in one of the weakest possible models -- robots do not agree on any coordinate system or compass, are \emph{oblivious} (have no memory of the past) and have only a local view.

\subsection{Model \& Time Notions}

We consider a set $\mathcal{R}$ of $n$ robots  $r_1, \dots ,r_n$, each of which occupies a single point in $\mathbb{R}^{3}$ at each time.
As such, robots can neither block each other's views nor paths, and multiple robots are allowed to occupy the same position at the same time.
The position of robot $r_i$ at time $t$ is denoted by $p_i(t)$.
The positions of all robots at time $t$, $\mathcal{P}_t = \big(p_1(t), \dots, p_n(t)\big)$ are collectively called the \emph{configuration} at time $t$.
The Euclidean distance between points $x,y \in \mathbb{R}^3$ is denoted as $d(x,y)$.
For a subset of the three-dimensional Euclidean space $\mathcal{P} \subseteq \mathbb{R}^3, d(x,\mathcal{P})$  is used as a shorthand for $\min_{y \in \mathcal{P}} d(x, y)$.

The overall abilities of the robots are rather limited:
They are not allowed to communicate with each other, they are \emph{identical} (they cannot be distinguished) and are \emph{oblivious}, meaning they have no memory of the past.
Furthermore they do not share a common coordinate system or orientation.
Robots are only able to observe the space around them within a limited viewing range of $1$, i.e.\ a robot $ r_i$ can see the position of another robot $r_j$ if and only if $d(p_i(t),p_j(t)) \leq 1$.
Two robots $r_i$ and $r_j$ with $d(p_i(t),p_j(t)) \leq 1$ are also called \emph{neighbors}.
The set of all neighbors of $r_i$ at time $t$ is called the \emph{neighborhood} of $r_i$ and is denoted as $\mathcal{R}_i(t)$.
This limited viewing range can also be considered to induce a unit ball graph $\ubg{t} = (\mathcal{R},E_t)$ at time $t$, whose nodes consist of the robots and where the set of edges $E_t$ contains an edge $\{r_i,r_j\}$ if and only if $d(p_i(t),p_j(t)) \leq 1$.
This graph is also called the \emph{visibility graph} at time $t$.
Note that the \ubgPlain{} is a generalization of  the two-dimensional \emph{unit disk graph} ($\mathrm{UDG}$) to three dimensions.

Starting from a configuration of $n$ robots in the three-dimensional Euclidean space that is connected at time $0$, i.e.\ \ubg{0} is connected, the goal is to gather all robots in one point.
This problem will be referred to as the (three-dimensional) \Gathering/ problem.
Note that the eventual gathering point is not predefined and can instead be chosen by the robots at runtime.
This also imposes a subgoal during the execution of any algorithm that solves this problem:
It has to be ensured that \ubg{t} remains connected.
Otherwise, the limited viewing range of the robots, combined with the fact that they do not share coordinate systems, makes it impossible for any deterministic algorithm to restore connectivity and the robots can no longer converge to the same point \cite{localgathering}.

Througout this work, we consider two different notions of time:
The fully synchronous \fsync{} model and the continuous time model.

\paragraph{\textbf{\fsync{}}:}
In \fsync{}, all robots operate in fully synchronous \texttt{Look-Compute-Move} (\LCM/) cycles.
In the \look/ phase, a  robot $r_i$ observes its environment, detects the set of all visible robots $\mathcal{R}_i(t)$ and stores a snapshot in its local memory.
Based on this snapshot, $r_i$ computes a target point in the \compute/ phase.
Finally, in the \move/ phase, $r_i$ moves to that target point.
The execution of a single \LCM/ cycle is also denoted as one \emph{round}.

\paragraph{\textbf{Continuous Time Model}:}
Generally, the continuous time model can be seen as a continuous variant of  \fsync{}, in which
robots only move an infinitesimal small distance towards their target points \cite{conf/antsw/GordonWB04}.
At every point in time, the movement of each robot $r_i$ can be expressed by a \emph{velocity vector} $\vec{v}_i(t)$ with $0 \leq \|\vec{v}_i(t)\| \leq 1$, i.e.\ the maximal speed of a robot is bounded by $1$.
In contrast to \fsync{}, the function $p_i \colon \mathbb{R}_{> 0} \to \mathbb{R}^3$, representing the position of $r_i$ at time $t$, is a continuous function and also called the \emph{trajectory} of $r_i$.
Although the trajectories are continuous, they are not necessarily differentiable because robots are able to change their speed and direction non-continuously.
However, natural movement strategies have (right) differentiable trajectories.
Thus, the velocity vector of a robot $\vec{v}_i \colon \mathbb{R}_{> 0} \to \mathbb{R}^3$ can be seen as the (right) derivative of $p_i$.

\subsection{Our Contribution}
The contribution of this paper is twofold.
We consider the fully synchronous \Fsync/ model and the continuous time model.
For \Fsync/, we introduce the strategy  \gtcThreeD/ (\gtcThreeDShort/), which is the three-dimensional generalization of  \gtc/ (\gtcShort/), invented for robots operating in the Euclidean plane \cite{localgathering}.
The main idea of \gtcThreeDShort/ is that robots move towards the center of the smallest enclosing sphere of all robots within their viewing radius, while ensuring that the configuration stays connected.
We prove a runtime bound of $\Theta \left(n^2\right)$ for \gtcThreeDShort/ which matches the runtime of the two-dimensional \gtcShort/ strategy.

For the continuous time model, we generalize the class of contracting strategies \cite{DBLP:conf/algosensors/LiMHP17} to three dimensions.
In contracting strategies, every robot that lies on the convex hull of all robots' positions moves always with speed $1$ into a direction that points inside or on the boundary of the convex hull.
We prove that every (three-dimensional) contracting gathering strategy gathers all robots on a single point in time at most $\mathcal{O}\left(\Delta \cdot n^{3/2}\right)$, where $\Delta$ denotes the (geometric) diameter of the initial configuration, i.e. the maximum Euclidean distance between any pair of robots.
This runtime bound differs from the runtime bound for two-dimensional contracting strategies by a factor of $\sqrt{n}$.
The lower bound is $\Omega \left(\Delta \cdot n\right)$ and already holds for the two-dimensional case~\cite{DBLP:conf/algosensors/LiMHP17}.
The main open question is whether $\mathcal{O}\left(\Delta \cdot n^{3/2}\right)$ is tight or can be improved to $\mathcal{O}\left(\Delta \cdot n\right)$.

Note that a contracting strategy is not necessarily local.
Therefore, we finally present two local, contracting strategies.
Our first example is
the continuous variant of \gtcThreeDShort/, called \gtcThreeDContShort/.
We prove that the strategy is contracting and thus gathers the robots in time $\mathcal{O}\left(\Delta \cdot n^{3/2}\right)$.
In addition, we present the class
of \emph{tangential-normal} strategies.
These strategies are local and maintain connectivity.
As an example for a strategy that is both tangential-normal and contracting,  we introduce the \moveOnAngleMinimizer/ strategy.
\subsection{Related Work}

In this overview over related work, we focus on the \Gathering/ problem for synchronized robots with local visibility in the Euclidean plane.
Beyond that, we give a summary about research concerning robot coordination problems in the three-dimensional Euclidean space.
For other models and coordination problems, which involve, among others, less synchronized schedulers or robots with a global view, we refer the reader to the recent survey \cite{series/lncs/11340}.

Ando, Suzuki and Yamashita introduced the \gtcShort/-stra\-te\-gy for fully synchronous robots with local view \cite{localgathering}.
In \gtcShort/, every robot moves in every round towards the center of the smallest enclosing circle of all robots within its viewing range while ensuring that the swarm remains connected.
Ando et al.\ could prove that \gtcShort/ solves the \Gathering/-problem in finite time.
Later on, Degener et al.\ could prove a tight runtime bound of $\Theta \left(n^2\right)$ for \gtcShort/ \cite{DBLP:conf/spaa/DegenerKLHPW11}.
By now, this is the best known runtime bound for a strategy that solves \Gathering/ of robots with local visibility and without agreement on any coordinate system or compass in \Fsync/.

Faster runtimes could so far only be obtained under different assumptions -- for example by introducing one-axis agreement or changing the time model.
Poudel and Sharma proved that it is possible to gather a swarm of robots with local view in time $\mathcal{O}\left(\Delta\right)$, where $\Delta$ denotes the diameter of the initial configuration~\cite{DBLP:conf/sss/PoudelS17}.
The main assumption for their strategy is that the robots agree on one axis of their coordinate systems.


The second time model we consider in this paper is the continuous time model, introduced by Gordon et al.\ \cite{conf/antsw/GordonWB04}.
In this time model robots do not operate in synchronized rounds but continuously observe their environment and move while having a bounded maximal speed.
Gordon et al.\ propose a gathering strategy for the continuous time model.
In their strategy, all robots that locally assume that they are located on the global convex hull move with maximal speed along the bisector formed by vectors to their neighbors along the global convex hull.
This strategy has later been called \textsc{Move-On-Bisector} by Degener et al.
They could also prove runtime of $\Theta \left(n\right)$ \cite{journals/topc/DegenerKKH15}.

The main result of this paper is based on a more general view on continuous \Gathering/ strategies in the Euclidean plane -- the class of \emph{contracting} strategies in which all robots that are located on the global convex hull of all robots move with maximal speed into a direction that points inside of the convex hull \cite{DBLP:conf/algosensors/LiMHP17}.
Li et al.\ could prove a runtime of $\mathcal{O} \left(\Delta \cdot n\right)$ for \emph{any} contracting strategy.
Note that \textsc{Move-On-Bisector} is also a contracting strategy but has a significantly faster runtime than $\mathcal{O} \left(\Delta \cdot n\right)$.
However, there are contracting strategies with a runtime of $\Omega \left(\Delta \cdot n\right)$ \cite{DBLP:conf/algosensors/LiMHP17}.

In the three-dimensional Euclidean space there is so far, to the best of our knowledge, no strategy known that solves \Gathering/ of robots with limited viewing range.
More generally, literature about robots operating in three-dimensional spaces is very scarce.
We summarize the literature briefly.
In \cite{DBLP:conf/icdcn/BhagatCM18} the authors show that gathering of robots in the three-dimensional Euclidean space is possible -- under the assumptions that robot have a global view but are not transparent and that the robots agree on one axis of their coordinate systems.
Tomaita et al.\ study a different problem -- the plane formation problem \cite{DBLP:conf/opodis/TomitaYKY17}.
In the plane formation problem, the goal is that eventually all robots are located on the same plane, while ensuring that no two robots occupy the same position.
The authors show that this problem is not solvable for every initial configuration, give a characterization of all start configurations for which the problem is solvable and introduce an algorithm that solves the problem for the latter set of configurations.
Yamauchi, Uehara and Yamashita generalize this result further and study the more general \patternForm/ for synchronized robots in the three-dimensional Euclidean space \cite{DBLP:conf/podc/YamauchiUY16}.
They characterize the set of all patterns that might be formable depending on symmetries of the initial configuration.


\section{Gathering in \fsync{}}
\label{sec:discrete}

In this section, the three-dimensional \Gathering/ problem will be studied under the \fsync{} model.
The results can be considered as a generalization of those obtained by Degener et al. \cite{DBLP:conf/spaa/DegenerKLHPW11} for the two-dimensional setting.
It will be shown that a generalization of \gtcShort/ by Ando et al. \cite{localgathering} solves the gathering problem in three dimensions in $\Theta(
n^2)$ rounds.

\subsection{\gtcThreeD/}

\begin{algorithm}
	\caption{\gtcThreeD/ (\gtcThreeDShort/)}
	\label{alg:gttc}
	\begin{algorithmic}[1]
		\State \( \mathcal{R}_i(t) := \{\text{positions of robots visible from } r_i \text{, including } r_i \text{ at time } t\} \)
		\State \( \mathcal{S}_i(t) := \text{smallest enclosing sphere of } \mathcal{R}_i(t) \)
		\State \( c_i(t) := \text{center of } \mathcal{S}_i(t) \)
		\Comment target point
		\ForAll{$r_j \in \mathcal{R}_i(t)$}
		\Comment Maintain connectivity
		\State \( m_j : = \) midpoint between $p_i(t)$ and $p_j(t)$
		\State \(\mathcal{B}_j(t) := \text{ball with radius } \frac{1}{2} \text{ and center } m_j \)
		\State \( \ell_j := \) maximum distance $r_i$ can move towards $c_i(t)$ without leaving $\mathcal{B}_j(t)$
		\EndFor
		\State \( L_i := \min_{r_j \in \mathcal{R}_i(t)} \ell_j \)
		\State Move towards $c_i(t)$ for a distance of $L_i$
	\end{algorithmic}
\end{algorithm}

The strategy \gtcThreeD/ (\gtcThreeDShort/) is a generalization of \gtc/ to the three-dimensional Euclidean space and is summarized in Algorithm \ref{alg:gttc}.
A key component is the computation of a \ac{ses} of a set of points $\mathcal{P}$.
This is a sphere of minimal radius that contains all points in $\mathcal{P}$ with the following properties:

\begin{restatable}{proposition}{smallestEnclosingSpheres}\cite{10.1287/mnsc.19.1.96}
	\label{prop:conv_comb}
	Let $\mathcal{S}$ be the smallest enclosing $d$-sphere (SES) of a point set $\mathcal{P} \subset \mathbb{R}^d$.
	Then the center $c$ of $\mathcal{S}$ is a convex combination of at most $d+1$ points in $\mathcal{P}$ that lie on the surface of $\mathcal{S}$.
	Especially, \begin{enumerate}
		\item $c$ lies in $S$
		\item $c$ minimizes the maximum distance to the points in $\mathcal{P}$.
	\end{enumerate}
\end{restatable}

Intuitively, \gtcThreeDShort/ works by attempting to locally move robots closer together.
This is achieved by letting each robot $r_i$ compute the \ac{ses} of its neighborhood $\mathcal{R}_i(t)$ and then moving towards its center $c_i(t)$.
Additionally, the strategy follows the subgoal of maintaining connectivity of \ubg{t+1}.
This is achieved by limiting the distance a robot $r_i$ moves towards its target $c_i(t)$, such that for any of its neighbors $r_j$, it stays within a distance of $\frac{1}{2}$ of the midpoint between the positions of $r_i$ and $r_j$ at time $t$.
Thus, if both $r_i$ and $r_j$ perform this strategy, the distance between their positions at the start of the next round $t+1$ is at most $1$, maintaining visibility.
By the argumentation above, the following Lemma holds.

\begin{lemma}
	If \ubg{0} is connected, \ubg{t} remains connected for all $t \geq 0$.
\end{lemma}

Overall, the only difference to the original \gtcShort/ strategy for two dimensions lies in the computation of a smallest enclosing sphere in the 3D case over a smallest enclosing circle in the 2D case.
In fact, if the three-dimensional version is applied to a configuration of robots that is coplanar with respect to some plane $h$, it acts just as if the robots' positions were projected to $h$ and the two-dimensional version was applied to the resulting two-dimensional subspace.
This is a result of the fact that computing a \ac{ses} of a set of coplanar points is equivalent to computing a smallest enclosing circle instead.

From this observation, we can immediately conclude that the lower bound on the runtime of the two-dimensional version of the strategy shown by Degener et al. \cite{DBLP:conf/spaa/DegenerKLHPW11} also applies to the three-dimensional case by simply embedding the two-dimensional worst-case start configuration within three-dimensional space:
In the configuration, $n$ robots are positioned on a circle such that the distance between two neighbors is $1$.
This causes the robots to only take small steps of size $\mathcal{O}(1/n)$ towards the center of the circle, leading to a gathering time of $\Omega(n^2)$.

\begin{theorem}
	\label{Thm:lower_bound}
	There is a start configuration such that \gtcThreeDShort/ takes $\Omega(n^2)$ rounds to gather the robots in one point.
\end{theorem}

With a generalization of the analysis of \cite{DBLP:conf/spaa/DegenerKLHPW11}, we can also prove an upper runtime bound of $\mathcal{O}\left(n^2\right)$.
Due to space constraints, the analysis is moved to \Cref{section:3dGtcAnalysis}.

\begin{restatable}{theorem}{mainTheorem}
	\label{Thm:disc_runtime}
	Given $n$ robots in a connected starting configuration $\mathcal{P} \in \mathbb{R}^3$ in the Euclidean space, \gtcThreeDShort/ gathers the robots in $\mathcal{O}(n^2)$ rounds.
\end{restatable}

The combination of both theorems yields a tight runtime of $\Theta(n^2)$.

\section{Continuous Gathering}

Now, we consider the \Gathering/ problem within the continuous time model.
For the Euclidean plane, Li et al. \cite{DBLP:conf/algosensors/LiMHP17} introduced the class of \emph{contracting strategies}.
This definition can also be applied to three dimensions:
Let $\ch_t$ denote the closed convex hull of the robots' configuration $\mathcal{P}_t$ at time $t$ and let $\corn_t$ denote the vertices of $CH_t$.
The class of contracting strategies can be defined as follows:

\begin{definition}
	\label{def:contracting}
	In the continuous time model, a movement strategy for $n$ robots is called \emph{contracting} if for every time $t$ such that the cardinality of $\corn_t$ is strictly greater than $1$, every robot in $\corn_t$ moves with speed $1$ in a direction that points to $\ch_t$.
\end{definition}

The main idea of our analysis is to project the three-dimensional configuration (including the velocity vectors) to a two-dimensional plane.
The projected robots then perform something similar to a contracting strategy where they move towards the inside of the projected convex hull with varying speeds.
However, when looking at only a single projection plane, some velocity vectors might even have a length of $0$ in the projection at some points in time (in case the projection plane is chosen orthogonal to the velocity vector).
Thus, the analysis of Li et al.\ cannot be directly applied to the projection as this analysis assumes that all robots on the convex hull move with speed $1$ towards the inside.
Instead, we analyze not only one but all possible (meaningfully different) projections, since -- intuitively -- for a majority of all possible projection planes, the projected length of a velocity vector must be larger than a constant $\varepsilon$.

\subsection{Preliminaries}

The following lemma is a useful tool for the analysis of continuous strategies stating how the distance between two robots changes over time.

\begin{lemma} [\cite{Kling0219}]
	\label{lemma:cont_dist_change}
	Consider two robots $r_i$ and $r_j$ with differentiable trajectories at time $t$.
	Their distance $d(p_i(t), p_j(t))$ at time $t$ changes with speed
	\begin{align*}
	d'(p_i(t), p_j(t)) = -(\|v_i(t)\| \cdot  \cos \beta_{i,j}(t) + \|v_j(t)\| \cdot \cos \beta_{j,i}(t)),
	\end{align*}
	where $\beta_{i,j}(t)$ is the angle between $\vec{v}_i(t)$ and the line segment $\overline{p_i(t)p_j(t)}$.
\end{lemma}

The main tool for the analysis of contracting strategies in the three-dimen\-sional Euclidean space are projections of the robots' configuration onto a two-dimensional plane.
Let $h(\vec{x})$ be the plane through the origin with normal vector $\vec{x}$ and let $\Pi_{\vec{x}}$ denote the orthogonal projection onto $h(\vec{x})$.
Now, given a configuration $\mathcal{P}$ of $n$ robots, consider their projection $\hat{\mathcal{P}}^{(\vec{x})} = \{ \Pi_{\vec{x}}p_i(t) \mid p_i(t) \in \mathcal{P} \}$ onto $h(\vec{x})$ along with the projections of their movement vectors $\hat{\vec{v}}^{(\vec{x})}_i(t) = \Pi_{\vec{x}}\vec{v}_i(t)$.
Furthermore, denote the convex hull of $\hat{\mathcal{P}}$ as $PCH_t(\vec{x})$.
See also Fig. \ref{fig:proj}.

\begin{figure}[tb]
	\label{fig:proj}
	\centering
	\includegraphics[width = 0.92\textwidth]{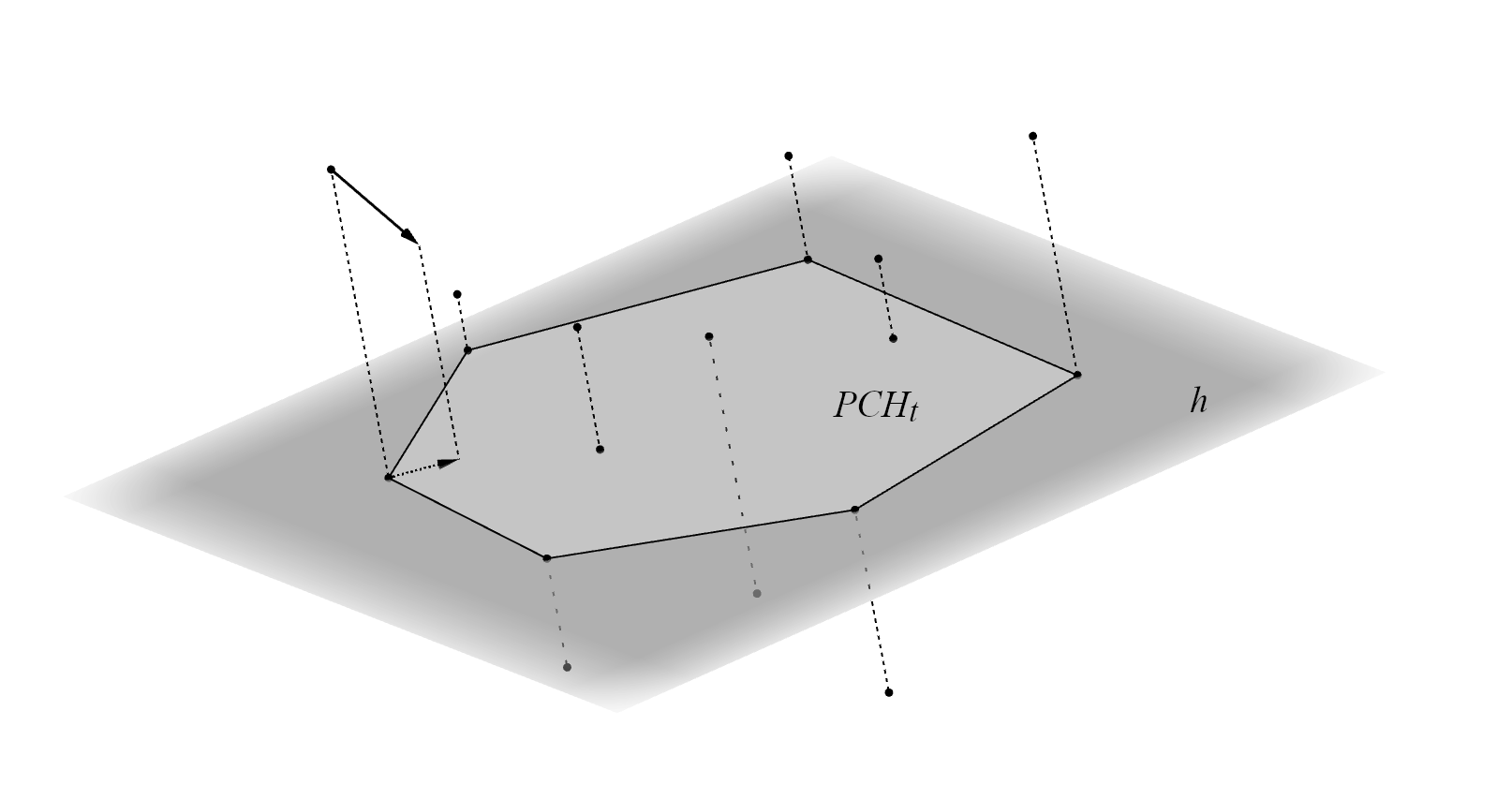}
	\caption{A configuration of robots being projected onto a plane $h(\vec{x})$. The mapping of the orthogonal projection $\Pi_{\vec{x}}$ is illustrated by dashed lines and the projected convex hull $\pch_t(\vec{x})$ is shown in light grey. One of the robots' movement vectors as well as its projection are represented by arrows.}
\end{figure}

If the robots perform a contracting strategy in the three-dimensional space, their projections also move towards the inside of the projected convex hull $PCH_t(\vec{x})$ since $\Pi_{\vec{x}}$ is a linear transformation and therefore preserves convexity.
However, the lengths of the projected movement vectors $\hat{\vec{v}}_i(t)$ are going be smaller than $1$ in general.
For a given projection onto a plane $h(\vec{x})$, the minimum length of the $\hat{\vec{v}}^{(\vec{x})}_i(t)$ will be called the \emph{projected speed} and is denoted by $\varepsilon_{\vec{x}} = \min_{r_i \in \mathcal{R}} ||\Pi_{\vec{x}}\vec{v}_i(t)||$.
Note that $\varepsilon_{\vec{x}}$ can even be $0$ in case $h(\vec{x})$ is orthogonal to any velocity vector.
The following notion of the \emph{length} of $PCH_t(\vec{x})$ will be used as a part of a progress measure for three-dimensional contracting strategies:

\begin{definition}(Length)
	Let \(m_1(t),m_2(t),...,m_{k(t)}(t)\) be the vertices of $\pch_t(\vec{x})$ (ordered counter-clockwise), where $k(t)$ is the number of vertices at time $t$.
	The \emph{length} $\ell(t,\vec{x})$ of $\pch_t(\vec{x})$ is defined as the sum of its edge lengths:
	$\ell(t,\vec{x}) = \sum_{\iota = 1}^{k(t)} d(m_\iota(t), m_{\iota - 1}(t)) $,
	where $m_0 := m_{k(t)}(t)$.
\end{definition}

Note that if the diameter of the starting configuration was $\Delta$, the length of a given projection can be at most $\pi \Delta$ (if it approximates a circle).
Furthermore, if $\ell(t, \vec{x}) = 0$, then the robots have either gathered in the original three-dimensional space or have formed a line that is parallel to $\vec{x}$.
In the latter case it only takes further time of at most $O(\Delta)$ for the robots to gather, as those robots that form the endpoints of the line have no choice but to move towards each other.
The following Lemma provides a statement about how the length changes over time.

\begin{restatable}{lemma}{contractingSpeedEpsilon}
	\label{Lemma:ell_bound}
	For time $t$, let $h(\vec{x})$ be a plane with projected speed $\varepsilon_{\vec{x}}$, such that $\ell(t, \vec{x}) > 0$ and no two robots with different positions in $\mathbb{R}^3$ get projected onto the same point on $h(\vec{x})$.
	Then $\ell'(t, \vec{x}) \leq -\frac{8\varepsilon_{\vec{x}}}{n}$.
\end{restatable}

\begin{proof}
	Because $\Pi_{\vec{x}}$ is a linear transformation, each of the $m_\iota(t)$ (corners of $\pch_t(\vec{x})$) must also be the projection of one of the vertices of the original, three-dimensional convex hull $\ch_t$.
	Therefore, they possess velocity vectors that point towards the inside of $\ch_t$ by the definition of a contracting strategy.
	Now consider the projections of these velocity vectors onto $h(\vec{x})$:
	Let $\hat{\vec{v}}^{(\vec{x})}_i(t) := \Pi_{\vec{x}} \vec{v}_i(t)$.
	By assumption, we have $||\hat{\vec{v}}^{(\vec{x})}_i(t)|| \geq \varepsilon_{\vec{x}}$.
	Using this, it is now possible to bound $\ell'(t, \vec{x})$:
	Let $\alpha_\iota(t)$ be the internal angle of $\pch_t(\vec{x})$ at $m_\iota(t)$.

	Note that in general, it may happen that two corner robots of $\ch_t$ got projected onto the same point on $h(\vec{x})$ for some $\vec{x}$.
	By one of the assumptions of the lemma, this is not true.
	Therefore, we know that each corner $m_\iota(t)$ of $\pch_t(\vec{x})$ contains only a single robot.
	This means that each $\alpha_\iota(t)$ is split into two parts, $\hat{\beta}_{\iota, \iota -1}(t)$ and $\hat{\beta}_{\iota - 1, \iota}(t)$ by $m_\iota(t)$'s velocity vector $\hat{\vec{v}}^{(\vec{x})}_\iota(t)$, such that $\alpha_\iota(t) = \hat{\beta}_{\iota, \iota -1}(t) + \hat{\beta}_{\iota - 1, \iota}(t)$.
	Using Lemma \ref{lemma:cont_dist_change} and Proposition \ref{prop:ineq_split_angle}, the derivative of $\ell(t)$ can now be bounded as follows:
	Recall that $\ell'(t,\vec{x})= \sum_{\iota = 1}^{k(t)} d'(m_\iota(t), m_{\iota - 1}(t))$:
	\begin{align}
	\ell'(t,\vec{x})&= \sum_{\iota = 1}^{k(t)} d'(m_\iota(t), m_{\iota - 1}(t))\\
	&= \sum_{\iota = 1}^{k(t)} -\big(||\hat{\vec{v}}^{(\vec{x})}_\iota(t)|| \cos \hat{\beta}_{\iota, \iota -1}(t) + ||\hat{\vec{v}}^{(\vec{x})}_{\iota-1}(t)|| \cos \hat{\beta}_{\iota-1, \iota}(t)  \big)\\
	&\leq - \varepsilon_{\vec{x}} \sum_{\iota = 1}^{k(t)} \cos \hat{\beta}_{\iota, \iota -1}(t) + \cos \hat{\beta}_{\iota-1, \iota}(t)\\
	&= - \varepsilon_{\vec{x}} \sum_{\iota = 1}^{k(t)} \frac{2(\alpha_\iota(t) - \pi)^2}{\pi^2} \label{equation:decomposition} \\
	&= - \frac{2\varepsilon_{\vec{x}}}{\pi^2} \sum_{\iota = 1}^{k(t)} (\alpha_\iota(t) - \pi)^2
	\end{align}
	For \Cref{equation:decomposition} observe that for $\vartheta \in [0,1]$ and $\alpha \in [0, \pi]$, it holds that $\cos(\alpha \vartheta) + \cos(\alpha (1 - \vartheta)) \geq \frac{2(\alpha - \pi)^2}{\pi^2}$ \cite{DBLP:conf/algosensors/LiHP16}.
	Now, the Cauchy-Schwarz inequality along with the fact that the sum of the inner angles of a convex polygon with $k$ corners is $(k - 2) \cdot \pi$.

	\begin{align*}
	\ell'(t, \vec{x})\leq - \frac{2 \varepsilon_{\vec{x}}}{k(t) \cdot \pi^2} \cdot \Big(\sum_{\iota = 1}^{k(t)} (\alpha_\iota(t) - \pi) \Big)^2 &= - \frac{2 \varepsilon_{\vec{x}}}{k(t) \cdot \pi^2} \cdot \big( (k(t) - 2) \cdot \pi - k(t) \cdot \pi \big)^2\\
	&= - \frac{8 \varepsilon_{\vec{x}}}{k(t)} \leq - \frac{8 \varepsilon_{\vec{x}}}{n}
	\end{align*}
	This concludes the proof. \qed
\end{proof}

Note that this also means that $\ell(t,\vec{x})$ is monotonically decreasing over time.

\subsection{Proof of the upper bound}
The main idea of the analysis is to track the lengths $\ell(t, \vec{x})$ for all (meaningfully different) projection planes $h(\vec{x})$.
Since the length of the normal vector does not matter, it is enough to consider only vectors $\vec{x}$ of length $1$.
Additionally, a vector $\vec{x}$ and its reflection about the origin $-\vec{x}$ describe the same plane.
Therefore it is enough to consider those vectors that lie on the surface of a unit hemisphere $U$ centered around the origin (w.l.o.g. the one above the $XY$-plane).

The integral of the lengths $\ell(t, \vec{x})$ with respect to $\vec{x}$ on the surface of $U$ at time $t$ can now be used as a measure to track the progress of a three-dimensional gathering strategy:
\begin{align*}
L(t) = \iint_U \ell(t, \vec{x}) dA
\end{align*}

If $L(t) = 0$, the robots have gathered.
If one of the $\ell(t, \vec{x})$ prematurely becomes $0$, then the robots are collinear and gather in further time $O(\Delta)$.

\begin{lemma}
	\label{Lemma:L_bound}
	$L(0) \leq 2 \pi^2 \Delta$.
\end{lemma}
\begin{proof}
	Since $\ell(t, \vec{x}) \leq \pi \Delta$ (if $\pch_t(\vec{x})$ approximates a circle), we conclude

	\begin{align*}
	L(0) &\leq \iint_U \pi \Delta dA = \pi \Delta \iint_U dA
	\end{align*}
	The remaining integral part is a surface integral over a hemisphere.
	By observing that the surface area of a unit hemisphere is $2\pi$, the lemma follows.
	\qed
\end{proof}

The goal of the proof is to show that there is at least a constant $(1-\alpha)$-fraction of projection planes $h(\vec{x})$ with projected speed at least $\varepsilon$ for some constants $\alpha$ and $\varepsilon$.
This can then be used to show that $L(t)$ decreases by a constant amount at each point in time using Lemma \ref{Lemma:ell_bound}.

\begin{figure}[tb]
	\label{fig:caps}
	\centering
	\def\svgwidth{0.6\linewidth}
	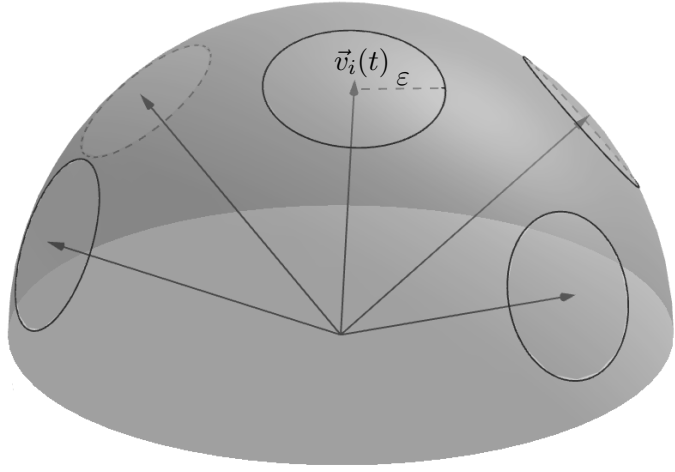
	\caption{A figure illustrating how movement vectors block areas of the unit hemisphere $U$. Around each movement vector $\vec{v}_i(t)$, there is a spherical cap of radius $\varepsilon$. Each plane corresponding to a normal vector $\vec{x}$ lying in one of those spherical caps is blocked.}
\end{figure}

Now consider a projection plane $h(\vec{x})$.
If this plane has projected speed smaller than $\varepsilon$ at time $t$, then there is a movement vector $\vec{v}_i(t)$, such that $\angle(\vec{x}, \vec{v}_i(t)) < \sin^{-1} \varepsilon$.
We say that $\vec{v}_i(t)$ \emph{blocks} $h(\vec{x})$.
Conversely, given a $\vec{v}_i(t)$, we can determine the set of all the $h(\vec{x})$ that are blocked by this $\vec{v}_i(t)$:

\begin{lemma}
	At time $t$, the movement vector $\vec{v}_i(t)$ blocks vectors from an area of $2 \pi \big(1 - \sqrt{1-\varepsilon^2} \big)$ on $U$ from reaching projected speed $\varepsilon$.
\end{lemma}
\begin{proof}
	W.l.o.g. it can be assumed that $\vec{v}_i(t)$ has a positive $z$-component, i.e. lies on $U$.
	Otherwise it can be reflected about the origin and it will still affect the exact same planes.

	Now consider the spherical cap of $U$ with base radius $\varepsilon$ and apex $\vec{v}_i(t)$ and let $C$ be its curved surface (see Fig. \ref{fig:caps} for an illustration).
	For all vectors $\vec{x} \in C$, $h(\vec{x})$ is blocked from reaching projected speed $\varepsilon$.
	The area of $C$ can be computed by $A_C = 2 \pi r^2 (1 - \cos \theta) = 2 \pi (1 - \cos (\sin^{-1} \varepsilon))
	= 2 \pi (1 - \sqrt{1 - \varepsilon^2})
	$
	\qed
\end{proof}

Since there are $n$ robots, the area blocked by their movement vectors is at most $n \cdot 2\pi (1 - \sqrt{1 - \varepsilon^2})$, whereas the total surface of $U$ is $2\pi$.
If we want the movement vectors to block only an $\alpha$-fraction of $U$'s surface, the $\varepsilon$ can be chosen accordingly:

\begin{lemma}
	\label{Lemma:alpha_eps_choice}
	Let $0 \leq \alpha \leq 1$. Then for a minimum speed of $\varepsilon = \frac{\sqrt{2n\alpha - \alpha^2}}{n}$, there is at most an $\alpha$-fraction of the surface of $U$ that is blocked with respect to $\varepsilon$.
\end{lemma}
\begin{proof}
	$U$ has a surface of $2\pi$ and the robots' movement vectors block an area of at most $n \cdot 2\pi (1 - \sqrt{1 - \varepsilon^2})$.
	We want to choose $\varepsilon$ such that the following holds:
	\begin{align*}
	\alpha 2\pi = n \cdot 2\pi (1 - \sqrt{1 - \varepsilon^2})
	\iff \;  \varepsilon = \frac{\sqrt{2n\alpha - \alpha^2}}{n}
	\end{align*}
	\qed

\end{proof}

Using this lemma, it is now possible to bound the decrease of the progress measure $L(t)$ for a given $\alpha$:

\begin{lemma}
	\label{Lemma:L_decrease}
	For a time $t \geq 0$ such that $\ell(t,\vec{x}) > 0$ for all $\vec{x} \in U$ and $0 \leq \alpha \leq 1$, then $L'(t) \leq - 16\pi \cdot (1-\alpha) \cdot \frac{\sqrt{2n\alpha - \alpha^2}}{n^2}$.
\end{lemma}
\begin{proof}
	Choose $\varepsilon = \frac{\sqrt{2n\alpha - \alpha^2}}{n}$ according to Lemma \ref{Lemma:alpha_eps_choice}, i.e. there is only at most an $\alpha$-fraction of the surface of $U$ that is blocked.
	Since Lemma \ref{Lemma:ell_bound} only applies to those $\vec{x}$ for which no two robots get projected onto the same point, the $\vec{x}$ for which this is the case still have to be considered.
	However, there is only a finite number $\binom{n}{2}$ of such vectors out of the uncountably many that form $U$ and they are only singular points on $U$.
	Therefore, they can be ignored when considering the integral $L(t)$.
	By Lemma \ref{Lemma:ell_bound}, there is an $(1 - \alpha)$-fraction of vectors $\vec{x}$ from the surface of $U$ (which has size $2\pi$) with
	$\ell'(t, \vec{x}) \leq -\frac{8\varepsilon}{n} = -8 \frac{\sqrt{2n\alpha - \alpha^2}}{n^2}$.
	Using this, we can bound $L'(t)$:
	\begin{align*}
	L'(t) = \frac{d}{dt} \Big(\iint_U \ell(t, \vec{x})dA \Big) = \iint_U \ell'(t, \vec{x})dA\\
	\leq  (1-\alpha) \cdot 2\pi \cdot -8 \frac{\sqrt{2n\alpha - \alpha^2}}{n^2}
	= -16\pi \cdot (1-\alpha) \cdot \frac{\sqrt{2n\alpha - \alpha^2}}{n^2}
	\end{align*}
	\qed
\end{proof}

By choosing the $\alpha$ appropriately, the main result can now be obtained:

\begin{theorem}
	A set of $n$ robots controlled by a contracting strategy gathers in time $\mathcal{O}\left(\Delta \cdot n^{3/2}\right)$ from an initial configuration with diameter $\Delta$.
\end{theorem}
\begin{proof}
	By Lemma \ref{Lemma:L_bound}, we have $L(0) \leq 2 \pi^2 \Delta$.
	By Lemma \ref{Lemma:L_decrease}, $L(t)$ decreases by at least $16\pi \cdot (1-\alpha) \cdot \frac{\sqrt{2n\alpha - \alpha^2}}{n^2}$ for a given $\alpha$ as long as $\ell(t, \vec{x}) > 0$ for all $\vec{x} \in U$.
	However, if there is an $\vec{x} \in U$ with $\ell(t, \vec{x}) = 0$, then the robots are collinear along some line that is parallel to $\vec{x}$ and take further time $O(\Delta)$ to gather.

	Now choose $\alpha = \frac{1}{2}$ and consider an arbitrary time $t$ such that $\ell(t, \vec{x}) > 0$ for all $\vec{x} \in U$.
	Then $L'(t) \leq -8\pi \frac{1}{n^{3/2}}$.
	Therefore it takes time at most $(2 \pi^2 \Delta) / (8\pi \frac{1}{n^{3/2}}) = \frac{\pi}{4} \Delta n^{3/2}$ until $L(t)$ is zero.
	This leads to a gathering time of $\mathcal{O}\left(\Delta \cdot n^{3/2}\right) + \mathcal{O}(\Delta) \in \mathcal{O}\left(\Delta \cdot n^{3/2}\right)$.
	\qed
\end{proof}

\subsection{\gtcThreeDCont/}
Next, a continuous version of \gtcThreeDShort/  which was already presented for the discrete time setting, will be considered as a concrete example of a contracting strategy.
The two-dimensional version of this strategy was adapted for continuous time by Li et al.\ \cite{DBLP:conf/algosensors/LiHP16}.
Compared to the discrete time version, no additional measures have to be taken to preserve connectivity, as it can be shown that this happens naturally in the continuous case.
The strategy is summarized in Algorithm \ref{alg:gttc_cont}.

\begin{algorithm}
	\caption{\gtcThreeDCont/ (\gtcThreeDContShort/)}
	\label{alg:gttc_cont}
	\begin{algorithmic}[1]
		\State \( \mathcal{R}_i(t) := \{\text{positions of robots visible from } r_i \text{, including } r_i \text{ at time } t\} \)
		\State \( \mathcal{S}_i(t) := \text{smallest enclosing sphere of } \mathcal{R}_i(t) \)
		\State \( c_i(t) := \text{center of } \mathcal{S}_i(t) \)
		\State Move towards $c_i(t)$ with speed $1$, or stay on $c_i(t)$ if $r_i$ is already positioned on it.
	\end{algorithmic}
\end{algorithm}

To show that \gtcThreeDContShort/ is contracting,  it must first be verified that connectivity of the visibility graph $UBG_t = (\mathcal{R}, E_t)$ is maintained at all times.
The same reasoning that was used in the two-dimensional case by Li et al. \cite{DBLP:conf/algosensors/LiHP16} can also be applied here:

\begin{restatable}{lemma}{contThreeDGTCConnectivity}
	\label{alg:gttc_cont_connectivity}
	Let $\mathcal{R}$ be a set of robots in the three-dimensional Euclidean space that follows the \gtcThreeDContShort/ strategy.
	If $\{r_i, r_j\}$ is an edge in $UBG_t$ at time $t$, then $\{r_i, r_j\}$ is an edge in $UBG_{t'}$ at $t' \geq t$.
	Thus, \gtcThreeDContShort/ maintains the connectivity of $\ubg{t}$.
\end{restatable}

\begin{proof}
	Consider a robot $r_i$ with neighborhood $\mathcal{R}_i(t)$ at time $t$.
	Let $Q_i(t)$ be the intersection of the unit balls of all robots in $\mathcal{R}_i(t)$.
	Since the \ac{ses} of $\mathcal{R}_i(t)$ can have a radius of at most $1$ and contains all robots in $\mathcal{R}_i(t)$, its center $c_i(t)$ must lie in $Q_i(t)$.

	Consider some neighbor $r_j \in \mathcal{R}_i(t)$ of $r_i$ and assume that there is some future point in time $t' > t$, such that $d(p_i(t'), p_j(t')) > 1$, i.e. $r_i$ and $r_j$ are no longer neighbors.
	Since the movement of robots is continuous, there must be some time $t^* \in [t, t']$, for which $d(p_i(t^*), p_j(t^*)) = 1$.

	Now let $L$ denote the intersection of the unit balls of $r_i$ and $r_j$ at time $t^*$.
	Any point in $L$ is within distance at most $1$ of both $r_i$ and $r_j$.
	Furthermore $L$ is a superset of both $Q_i(t^*)$ and $Q_j(t^*)$, meaning the target points $c_i(t^*)$ and $c_j(t^*)$ of both $r_i$ and $r_j$ also lie in $L$.
	Therefore, $r_i$ and $r_j$ can only move in the direction of points that are in distance at most $1$ from both of them, meaning their distance can never exceed $1$, creating a contradiction to the assumption that their distance is greater than $1$ at time $t'$.
	\qed
\end{proof}

It remains to show that \gtcThreeDContShort/ is a contracting strategy.
This follows directly from \Cref{alg:gttc_cont_connectivity} and Proposition \ref{prop:conv_comb}, which states that the center of a \ac{ses} is a convex combination of the points it encloses, meaning any target point computed by the strategy lies within the convex hull of the current configuration.

\begin{theorem}
	\label{thm:gttc_cont_contracting}
	\gtcThreeDContShort/ is a contracting, local strategy and thus gathers the robots in time $\mathcal{O}\left(\Delta \cdot n^{3/2}\right)$.
\end{theorem}

\subsection{Tangential-Normal Strategies}
Previously, we showed a runtime bound for a relatively general class of (not necessarily local) gathering strategies and introduced a concrete example in \gtcThreeDContShort/.
However, when designing a local strategy, additional care has to be taken to maintain the visibility graph $\ubg{t}$ to successfully solve the \Gathering/ problem.
It would be useful to also have a relatively simple design criterion that ensures this property.
For this purpose, we will focus on robots' \emph{local convex hulls} and introduce the notion of \tangentialNormal/ strategies.
Let $\ch(\mathcal{R}_i(t))$ denote the local convex hull of robot $r_i$, i.e. the convex hull of $r_i$'s neighborhood.
Furthermore, let $\adj_t(i)$ denote the set of robots that are adjacent to $r_i$ on $\ch(\mathcal{R}_i(t))$ if $r_i$ lies on $\ch(\mathcal{R}_i(t))$ itself.
The main idea is to identify those velocity vectors that lead to a decrease in distance to all neighboring robots.
These vectors are the normal vectors of \emph{tangential planes}:

\begin{definition}
	\label{def:tang_plane}
	Given a convex polyhedron $P \subset \mathbb{R}^3$ and a vertex $p \in P$.
	A \emph{tangential plane} $h_p$ w.r.t. $P$ through $p$ is a plane that only intersects $P$ at the vertex~$p$.
\end{definition}

Note that as long as $P$ is actually convex, such a plane always exists and can -- for example -- be obtained by taking the plane through one of the faces adjacent to $p$ and slightly rotating it.
Based on this notion, we define the class of \emph{tangential-normal} strategies in which the corner robots of local convex hulls move along the normal vectors of tangential planes:

\begin{definition}
	\label{def:tang_normal}
	In the continuous time model, a gathering strategy for $n$ robots is called \emph{tangential-normal} if for every time $t$ in which the robots have not yet gathered, each robot $r_i \in \mathcal{R}$ that is on a corner of its own local convex hull $CH(\mathcal{R}_i(t))$ moves with speed $1$ along the normal vector of a tangential plane w.r.t. $CH(\mathcal{R}_i(t))$ through $p_i$ while other robots do not move.
\end{definition}


The following lemma characterizes the normal vectors of tangential planes and will be used to show the desired properties of tangential-normal strategies.


\begin{lemma}
	\label{lemma:tang_plane_angles}
	Let $p_i$ be a corner of a convex polyhedron $P$ and let $E_i$ be the set of edges of $P$ adjacent to $p_i$.
	Then a plane $h$ through $p_i$ with normal vector $\vec{n}$ is a tangential plane w.r.t. $P$ if and only if for each edge $e \in E_i$, $\angle(\vec{n}, e) < \frac{\pi}{2}$
\end{lemma}
\begin{proof}
	First, note that by the convexity of $P$ and since $h$ only intersects with it in $p_i$, the entire rest of $P$ lies on one side of $h$.
	However, if there was an edge $e$ with $\angle(\vec{n}, e) \geq \frac{\pi}{2}$, this would mean that $e$ lies on the opposite side of or directly on $h$, both of which are contradictions to $h$ being a tangential plane.

	For the other direction of the statement, let $\vec{n}$ be a vector such that for each edge $e \in E_i$, $\angle(\vec{n}, e) < \frac{\pi}{2}$.
	This property now immediately yields that all edges $e \in E_i$ lie on the same side of the plane $h: \vec{n} \cdot (x - p_i) = 0$ defined by $\vec{n}$ and the point $p_i$, making $h$ a tangential plane w.r.t. $P$ through $p_i$.
	\qed
\end{proof}

\begin{theorem}
	\label{thm:tang_normal_finite}
	Let $\mathcal{R}$ be a set of robots controlled by a tangential-normal strategy.
	Then, for each pair of robots $r_i, r_j \in \mathcal{R}$ and time $t$ such that $\{r_i, r_j\}$ is an edge in \ubg{t}, $\{r_i, r_j\}$ is an edge in \ubg{t'} for all $t' \geq t$.
	Thus, tangential-normal strategies maintain the connectivity of $\ubg{t}$.
\end{theorem}
\begin{proof}
	Let $\mathcal{R}$ be a set of robots that follows a tangential-normal strategy.
	Consider a time $t$ and a robot $r_i \in \mathcal{R}$ that lies on the corner of its own local convex hull $CH(\mathcal{R}_i(t))$ and let $v_i(t)$ be the normal vector of a tangential plane $h_{p_i}$ w.r.t. $CH(\mathcal{R}_i(t))$.
	By Lemma \ref{lemma:tang_plane_angles}, for each adjacent robot $r_j \in \adj_t(i)$, it holds that $\beta_{i,j}(t) < \frac{\pi}{2}$.
	Since the cosine is positive on the interval $[0, \frac{\pi}{2}]$, Lemma \ref{lemma:cont_dist_change} yields that $r_i$ contributes a strict decrease in distance to all of its neighbors.
	A neighbor $r_j$ now either does not move or also contributes a decrease in distance to $r_i$.
	This also means that for each pair of robots $r_i$ and $r_j$ that can see each other, the distance between $r_i$ and $r_j$ cannot increase, guaranteeing that the visibility graph remains connected.
	\qed
\end{proof}

Note however that while the \tangentialNormal/ property is a sufficient condition for ensuring connectivity, it is not a necessary condition.
In particular, \gtcThreeDContShort/ is not \tangentialNormal/ but is still able to maintain the connectivity of the visibility graph.

\subsection{\moveOnAngleMinimizer/}
Next, we introduce a strategy based on the tangential-normal criterion.
It is based around the idea to find a movement vector that somehow causes a large decrease in distance to all neighbors.
Since a smaller angle causes a greater decrease in distance (according to Lemma \ref{lemma:cont_dist_change}), one intuitive approach might be to find a movement vector that minimizes the maximal angle to all neighbors on the convex hull.
To this end, the notion of an \emph{angle minimizer} will be introduced.
Let $V = \{ \vec{v_1},\vec{v}_2,...,\vec{v}_k \} \subset \mathbb{R}^3$ be a set of vectors that lie on one side of a plane through the origin.
Then the vector $\vec{x}^* = \argmin_{\vec{x} \in \mathbb{R}^3} \max_{\vec{v}_i \in V} \angle(\vec{x}, \vec{v}_i)$ is called an angle minimizer of $V$.

Now, we define a strategy in which each robot that is a corner of its local convex hull $r_i$ moves along the angle minimizer of the edges between itself and the robots in $\adj_t(i)$.
This strategy will be called \moveOnAngleMinimizer/ and is summarized in Algorithm \ref{alg:mam}.

\begin{algorithm}
	\caption{\moveOnAngleMinimizer/}
	\label{alg:mam}
	\begin{algorithmic}[1]
		\State \( \mathcal{R}_i(t) := \{\text{positions of robots visible from } r_i \text{, including } r_i \text{ at time } t\} \)
		\State \( CH(\mathcal{R}_i(t)) := \text{Convex hull of } r_i \text{'s neighborhood} \)
		\If{$r_i$ is on a corner of $CH(\mathcal{R}_i(t))$}
		\State \( \vec{x}^* = \argmin_{\vec{x} \in \mathbb{R}^3} \max_{r_j \in \adj_t(i)} \angle(\vec{x}, p_j(t) - p_i(t)) \)
		\State $r_i$ moves along $\vec{x}^*$ with speed $1$
		\Else
		\State $r_i$ does not move
		\EndIf
	\end{algorithmic}
\end{algorithm}

Note that if a robot $r_i$'s local convex hull is two-dimensional, the angle minimizer is identical to the angle bisector of the inner angle at $r_i$.
Therefore, this strategy can also be viewed as a generalization of \mobs/ for two-dimensional continuous gathering \cite{conf/antsw/GordonWB04}, for which Kempkes et al.\ \cite{journals/topc/DegenerKKH15} could show an optimal gathering time of $\Theta(n)$.

It will now be shown that the presented strategy is both a tangential-normal and a contracting strategy.
By Lemma \ref{lemma:tang_plane_angles} and the existence of a tangential plane, we already know that there is a possible movement vector that has an angle of less than $\pi/2$ to all neighbors on the local convex hull
Therefore, the same must hold for $\vec{x}^*$, immediately showing that \moveOnAngleMinimizer/ is a tangential-normal strategy.

\begin{lemma}
	\label{lemma:mam_connectivity}
	\moveOnAngleMinimizer/ is a tangential-normal strategy.
\end{lemma}

\paragraph{\textbf{Computing $\vec{x}^*$}}
In order to see that \moveOnAngleMinimizer/ is also a contracting strategy, we look at a method to compute the angle minimizer.

Let $\hat{\vec{v}} = \vec{v} / ||\vec{v}||$ denote the respective normalized vector of $\vec{v}$ and let $\hat{V} = \{\hat{\vec{v}} \;|\; \vec{v} \in V\}$ for a set $V$ of vectors.
Then the following holds:


\begin{lemma}
	\label{lemma:angle_minimizer}
	Let $V \subset \mathbb{R}^3$ be a set of vectors that all lie on one side of a plane through the origin.
	The center $\vec{c}$ of the smallest enclosing sphere of $\hat{V}$ is an angle minimizer of $V$.
\end{lemma}
\begin{proof}
	Let $\vec{x}$ be a vector such that $\angle(\vec{x}, \vec{v}_i) \leq \pi/2$ for all $\vec{v}_i \in V$.
	Such a vector exists, since there is a plane such that all $\vec{v}_i$ lie on one side of this plane.

	Now consider the normalized vectors $\hat{\vec{v}}_i$.
	They lie on the surface of the unit sphere centred on the origin.
	Let $C_{\vec{x}}$ be the minimal spherical cap centred on the vector $\vec{x}$ such that all the $\hat{\vec{v}}_i$ lie on its surface.
	The vector $\hat{\vec{v}}_j \in \hat{V}$ with the maximal angle to $\vec{x}$ lies on the edge of the base of $C_{\vec{x}}$.
	The maximal angle can now be computed using the radius $r$ of $C_{\vec{x}}$ as \( \angle(\vec{x}, \hat{\vec{v}}_j) = \sin^{-1} r \).

	Since $\sin^{-1}$ is monotonically increasing on the interval $[0,1]$, finding the angle minimizer $\vec{x}^*$ now amounts to finding the center $\vec{c}$ of a spherical cap with minimal radius, which can be achieved by computing the smallest enclosing sphere of $\hat{V}$. \qed
\end{proof}

By applying the fact that the \ac{ses} of $\hat{V}$ is a convex combination of $\hat{V}$ (Proposition \ref{prop:conv_comb}), this lemma together with Lemma \ref{lemma:mam_connectivity} immediately yields that \moveOnAngleMinimizer/ is also a contracting strategy.

\begin{theorem}
	\label{thm:mam_contracting_tang_normal}

	\moveOnAngleMinimizer/ is a tangential-normal and a contracting strategy.
	Thus, it gathers the robots in time $\mathcal{O}\left(\Delta \cdot n^{3/2}\right)$.
\end{theorem}

%
%
%
\bibliographystyle{splncs04}
\bibliography{fullArXiv}

\newpage
\appendix
\section{Analysis of \gtcThreeD/} \label{section:3dGtcAnalysis}

In this section, we show that the upper bound of $\mathcal{O}(n^2)$ that Degener et al. \cite{DBLP:conf/spaa/DegenerKLHPW11} presented for \gtc/ also holds for \gtcThreeD/.
As such, the analysis largely consists of a generalization of the proof by Degener et al.\ for three dimensions and therefore, each lemma provides an analogous result to a similar lemma used in the original proof for two dimensions.

\begin{figure}
	\label{fig:spherical_cap}
	\centering
	\includegraphics[scale=1.0]{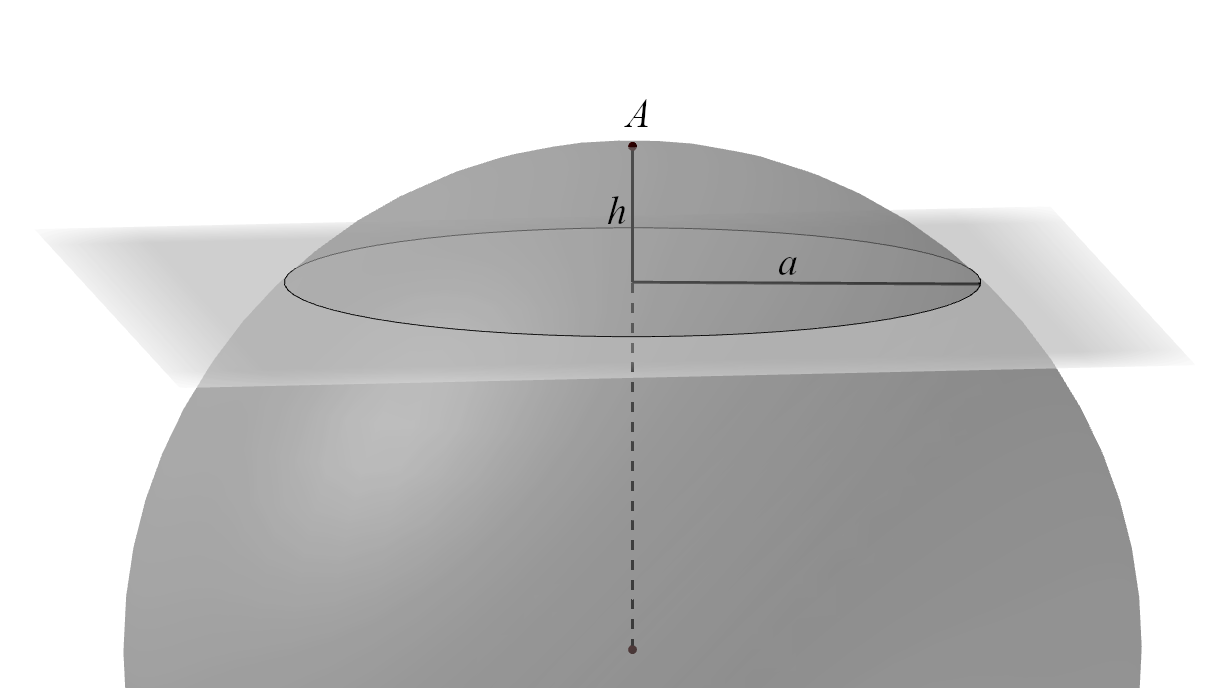}
	\caption{A plane dividing a sphere into two spherical caps. The smaller cap on top has apex $A$, height $h$ and base radius $a$.}
\end{figure}

Before beginning with the main analysis, some useful general properties of \ac{ses}s will be stated.
Two concepts that will be used throughout the proof are the notions of \emph{convex combinations} and \emph{spherical caps}.

Given a finite set of points \(\mathcal{P} = \{p_1,\dots,p_n\}\), a point $x$ is a convex combination of $\mathcal{P}$, if there are scalar coefficients $\alpha_i \geq 0$, such that $x = \sum_{i=1}^n \alpha_i p_i$ and $\sum_{i=1}^n \alpha_i = 1$.
The set of all convex combinations of $\mathcal{P}$ is identical to the convex hull of $\mathcal{P}$.
A \emph{spherical cap} is a region of a sphere cut off by a plane and is illustrated in Fig. \ref{fig:spherical_cap}.
It consists of a circular base with a radius $a$ which is formed by the intersection of the plane and the sphere and a curved surface with an apex $A$.
The distance between the base and the apex is called the height $h$ of the spherical cap.

\smallestEnclosingSpheres*

\begin{definition}
	\label{Def:pf_cap}
	Let $\mathcal{S}$ be a sphere and $\mathcal{P}$ a point set.
	A spherical cap of $\mathcal{S}$ whose curved surface does not contain any points from $\mathcal{P}$ is called a \emph{point-free cap}.
\end{definition}

Using this definition, the following Lemma can be stated about \ac{ses}s:

\begin{lemma}
	\label{Lemma:pf_cap_bound}
	Let $\mathcal{S}$ be the \ac{ses} with radius $r$ of a set of $n \geq 2$ points.
	Then there is no point-free cap of $\mathcal{S}$ with height $h$, such that $h > r$.
\end{lemma}
\begin{proof}
	Let $\mathcal{S}$ be the \ac{ses} of a set of $n \geq 2$ points with center $c$ and radius $r$.
	Assume there is a point-free cap with height $h > r$.

	Without loss of generality, rotate and translate the coordinate system such that $c$ lies at the origin and the apex of the point-free cap lies on the positive $x$-axis.
	Let $c_x$ be the $x$-component of $c$.
	Then for any point $p$ on the surface of $\mathcal{S}$, its $x$-component $p_x$ is at most
	\begin{align}
	p_x < c_x - (h - r) < c_x.
	\end{align}
	Therefore, $c$ cannot be a convex combination of points on the surface of $\mathcal{S}$.
	By Prop.~\ref{prop:conv_comb}, $c$ cannot have been the center of a \ac{ses}, leading to a contradiction.
	\qed
\end{proof}

\subsubsection{Progress Measures}

In order to determine the progress of the gathering, two measures will be used.
Firstly, consider what happens if two robots move to the exact same position at the same time:
Since the strategy is deterministic, the two robots will always observe the same neighborhood and therefore always compute the same target point.
We say two such robots have \emph{merged}.
Starting with $n$ robots, there can clearly be at most $n-1$ rounds with merging events until all robots have gathered.
This number of merges will be used as the first progress measure.

Now let $\mathcal{N}_t$ be the global \ac{ses} around a center $\mathbf{M}_t$ with radius $R_t$ in some round $t \geq 0$.
The radius $R_t$ can be used as a second progress measure:
The robots always move towards the centers of their respective (local) \ac{ses}.
By Prop. \ref{prop:conv_comb}, these centers always lie within the convex hulls of the respective neighborhoods and therefore also within the global convex hull.
Since robots do not leave the global convex hull, they also cannot leave $\mathcal{N}_t$, meaning $R_t$ can not increase with time.

On the other hand, once $R_t$ is smaller than $1/2$, all robots are within distance of at most $1$ of each other, meaning each robot sees every other robot.
At this point, they will all compute the same target point $\mathbf{M}_t$ and be able to move to it since the limiting spheres $\mathcal{B}_j(t)$ with radii $1/2$ must contain $\mathbf{M}_t$ for all pairs of robots $r_i$ and $r_j$.

The overall idea of the proof is now to show that during a constant number of rounds, either two robots merge, or $R_t$ decreases by at least $\Omega(1/n)$.
Since the initial radius $R_0$ is $\mathcal{O}(n)$ if the initial configuration is connected, this yields a total gathering time of $\mathcal{O}(n^2)$ rounds.

During the rest of the analysis, consider some arbitrary but fixed round $t_0$ and let $\mathcal{N}$ and $R$ denote $\mathcal{N}_{t_0}$ and $R_{t_0}$ respectively.
In order to show the progress of the two measures mentioned above, the analysis will focus on a certain region of $\mathcal{N}$.
Let $P$ be an arbitrary point on the surface of $\mathcal{N}$ and define a spherical cap $C$ with apex $P$ as follows (see Fig. \ref{fig:C_def} for an illustration of the construction):
Choose the height $h$ of $C$, such that the slant height of the inscribed cone of $C$ is $1/8$, i.e. such that $h^2 + a^2 = ( 1/8 )^2$.
Note that this causes the radius $a$ to be bounded by $1/8$.
Therefore the diameter of $C$ as a whole is at most $1/4$.
$C$ is then separated into two parts by intersecting $\mathcal{N}$ with another plane that is parallel to the base of $C$, creating another spherical cap $C_1$ with height $h/2$ that also has $P$ as its apex.
The remaining spherical segment $C \setminus C_1$ is called $C_2$.

The goal is to now show that within two consecutive rounds $t_0$ and $t_0 + 1$, either at least two robots merge, or all robots leave $C_1$.
In order to quantify the progress made in case the latter happens, a bound on $h$ is needed, which will be provided by the following lemma:

\begin{lemma}
	\label{Lemma:height_bound}
	The spherical cap $C$ has a height $h$ of at least \(\frac{1}{64 \pi \cdot R} \in \Omega\big( \frac{1}{n} \big)\).
\end{lemma}
\begin{proof}
	The proof proceeds very similarly to that of the analogous lemma by Degener et al. \cite{DBLP:conf/spaa/DegenerKLHPW11} for the two-dimensional case.
	The goal is to use the angle $\alpha$ between the height $h$ of $C$ and the (known) slant height ($1/8$) of the inscribed cone of $C$ (see Fig. \ref{fig:C_def}) to compute a bound on $h$.
	In order to do this, consider a circle $K$ with radius $R$, with the same midpoint as $\mathcal{N}$ and containing $P$.
	The angle $\alpha$ can be considered to be half of the internal angle of a regular convex polygon with side length $1/8$ and whose $m$ vertices lie on $K$.
	Since the circumference of $K$ is $2 \pi R$, the number of vertices $m$ is bounded by $16 \pi R$.
	The sum of angles of such a convex polygon is now given by $(m - 2) \pi$.
	This yields the following bound on $\alpha$:
	\begin{align*}
	2\alpha &\leq \frac{(16 \pi R - 2) \pi}{16 \pi R} = \pi - \frac{1}{8R}\\
	\Leftrightarrow \alpha &\leq \frac{\pi}{2} - \frac{1}{16R}
	\end{align*}

	Using the fact that \(\cos(x) \geq -\frac{2}{\pi}x + 1\) for \(0 \leq x \leq \pi/2\), we can now obtain a bound on $h$:

	\begin{align*}
	h &= \frac{\cos(\alpha)}{8} \geq \frac{1}{8} \cdot \Big( -\frac{2}{\pi}\alpha + 1 \Big)\\
	&\geq \frac{1}{8} \cdot \Big( - \frac{2}{\pi} \cdot \Big(\frac{\pi}{2} - \frac{1}{16R}\Big) + 1 \Big)\\
	&= \frac{1}{64 \pi R}
	\end{align*}

	Since $\mathcal{N}$ is a \ac{ses}, its radius $R$ is at most $n/2$.
	Therefore, we have \(h \in \Omega\big(\frac{1}{n}\big)\).
	\qed
\end{proof}

\begin{figure}[htb]
	\label{fig:C_def}
	\centering
	\includegraphics[width = 1 \textwidth]{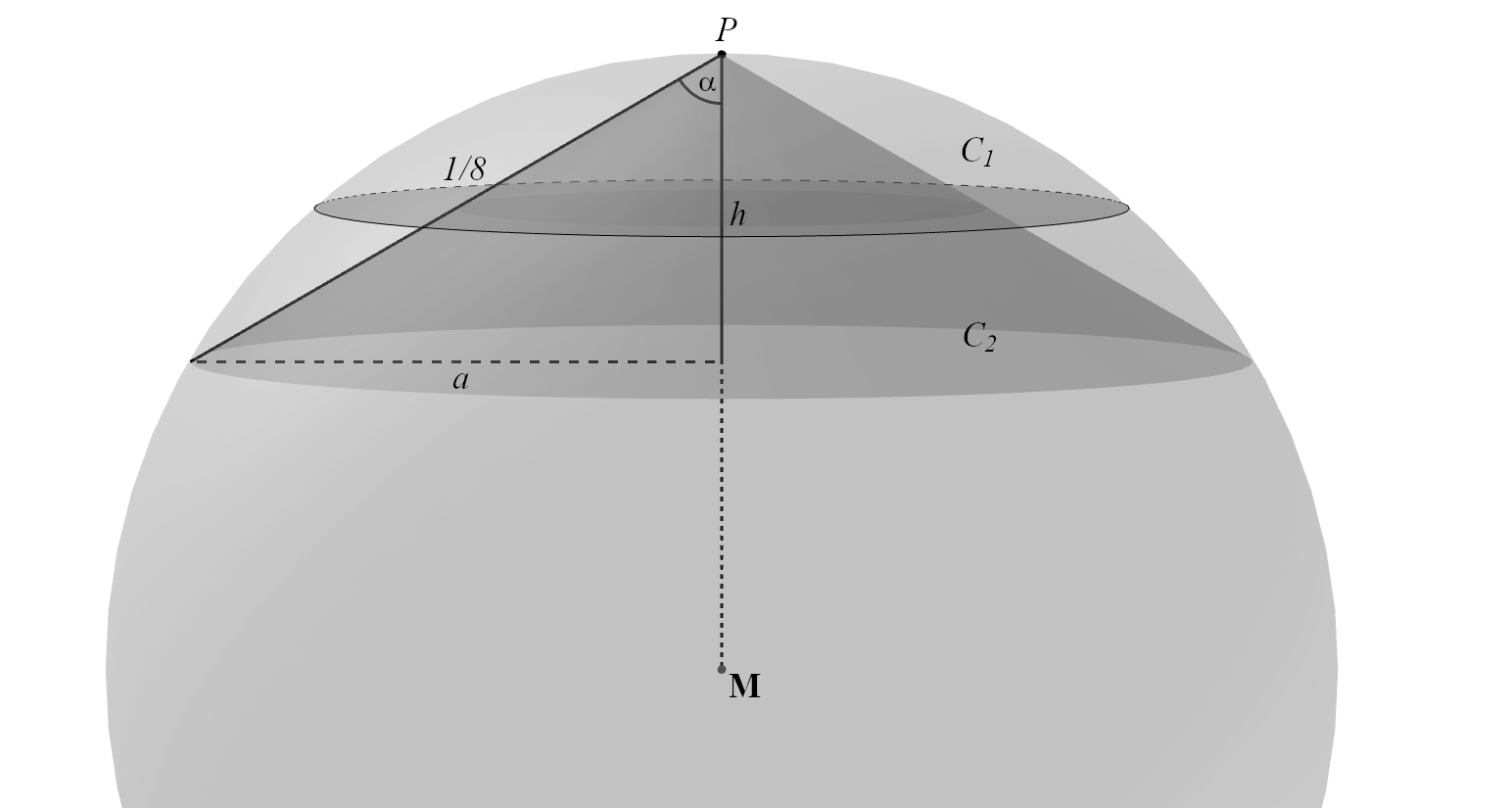}
	\caption{Illustration of the spherical cap $C$ with apex $P$, consisting of the smaller cap $C_1$ and the spherical segment $C_2$. The height $h$ is chosen such that the inscribed cone of $C$ has a slant height of $1/8$.}
\end{figure}

\subsubsection{Properties of Target Points}
First, some properties of \ac{ses} that can be computed by robots in and around $C$ will be given.
These will be useful later on to identify situations in which robots leave $C$, as well as ensuring that no further robots enter it.

\begin{figure}[tbh]
	\label{fig:SES_cap_bound}
	\centering
	\includegraphics[scale=1.2]{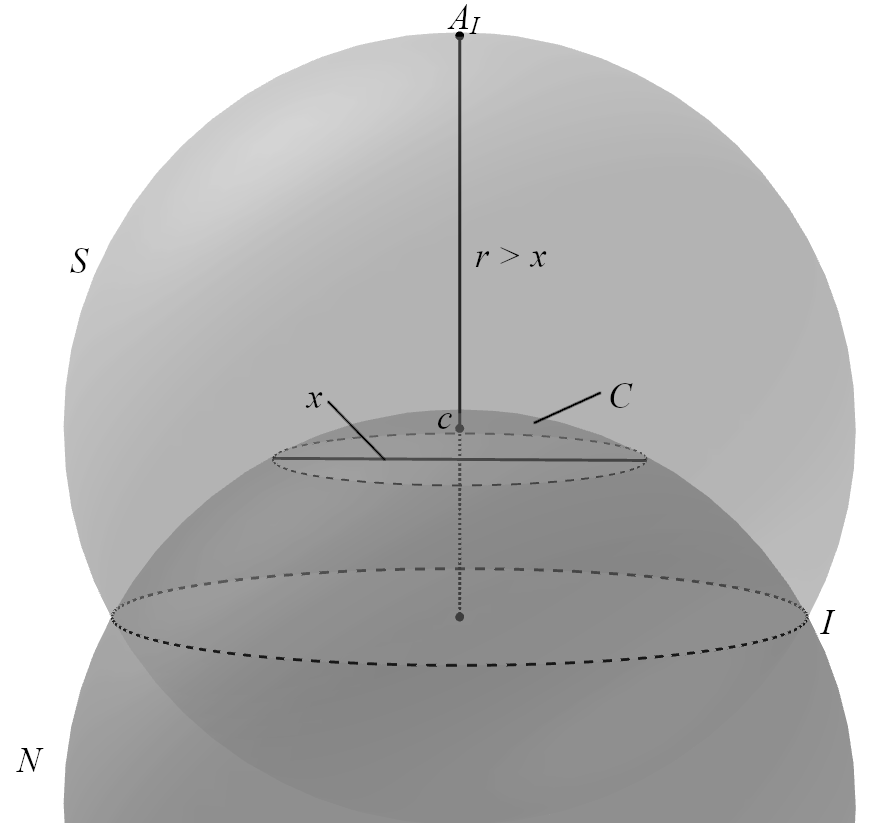}
	\caption{Illustration of the setting of Lemma \ref{Lemma:SES_cap_bound}. A sphere with its center $c$ in the spherical cap $C$ and a radius $r$ that is greater than the diameter of $x$ of $C$ cannot be a \ac{ses} of points in $\mathcal{N}$.}
\end{figure}

\begin{lemma}
	\label{Lemma:SES_cap_bound}
	Let $x$ be the diameter of a spherical cap $C$ of $\mathcal{N}$.
	Then any sphere $S$ with its center $c$ in $C$ and radius $r > x$ cannot be a \ac{ses} of points in $\mathcal{N}$.
\end{lemma}
\begin{proof}
	The setting of the lemma is described in Fig. \ref{fig:SES_cap_bound}.
	Since the center of $S$ lies in $C$ and the radius $r$ of $S$ is greater than the diameter $x$ of $C$, every intersection point of $S$ and $\mathcal{N}$ lies outside of $C$.
	The (circular) intersection of $S$ and $\mathcal{N}$, \(I = S \cap \mathcal{N}\), therefore lies completely outside of $S$.
	Now let $C_I$ be the spherical cap of $S$ whose curved surface lies outside of $\mathcal{N}$ and that has $I$ as its base.
	Let $h_I$ be its height and $A_I$ be its apex.
	Note that $C_I$ is therefore a point-free cap w.r.t. any set of points lying in $\mathcal{N}$.

	Now consider the height $h_I$ of $C_I$:
	Since $I$ lies outside of $C$, so does its center $c_I$.
	The line segment between $A_I$ and $c_I$ therefore has to pass through $c$, which lies within $C$.
	Thus, we have that $h_I > r$.
	By Lemma \ref{Lemma:pf_cap_bound}, this means that $S$ cannot have been a \ac{ses} of points in $\mathcal{N}$.
	\qed
\end{proof}

Since the diameter of $C$ is at most $1/4$, this immediately yields the following corollary:

\begin{corollary}
	\label{Cor:C_center_radius}
	The radius of a \ac{ses} of a point set $\mathcal{P} \subseteq \mathcal{N}$ with its center in $C$ is at most $1/4$.
\end{corollary}

The next lemma will be used to show that a robot with a neighbor that is positioned far away from $C$ cannot compute a target point in $C$.

\begin{lemma}
	\label{Lemma:dist_point}
	Let \(\mathcal{P} \subseteq \mathcal{N}\) be a set of points and let $\mathbf{A}$ be a point in $C$ (not necessarily in $\mathcal{P}$).
	If there is a point \(\mathbf{B} \in \mathcal{P}\) in distance more than $1/2$ from $\mathbf{A}$, then the center of the \ac{ses} of $\mathcal{P}$ cannot lie in the cap $C$.
\end{lemma}
\begin{proof}
	Assume the \ac{ses} $S$ of $\mathcal{P}$ has its center $c$ in $C$.
	Corollary \ref{Cor:C_center_radius} shows that the radius of $S$ can be at most $1/4$ and since $\mathbf{B}$ lies in $S$, we have \(d(c, \mathbf{B}) \leq 1/4\).
	Since $\mathbf{A}$ and $c$ both lie in $C$ (which has a diameter of at most $1/4$),
	we also have \(d(c, \mathbf{A}) \leq 1/4\).
	Applying the triangle inequality yields:
	\begin{align*}
	d(\mathbf{A}, \mathbf{B}) &\leq d(c, \mathbf{A}) + d(c, \mathbf{B}) \leq \frac{1}{4} + \frac{1}{4} = \frac{1}{2}
	\end{align*}
	Which is a contradiction to \(d(\mathbf{A}, \mathbf{B}) > 1/2\), meaning $c$ cannot lie in $C$.
	\qed
\end{proof}

The next lemma concerns local configurations of robots consisting of only a single robot in $C$ and a number of robots outside of $C$.
It shows that in such a scenario, the target point computed by the strategy must lie outside of $C_1$.

\begin{lemma}
	\label{Lemma:single_C1_point}
	The center of the \ac{ses} $\mathcal{S}$ of a non-empty point set \(\mathcal{P} \subseteq \mathcal{N} \setminus C \) and a point $\mathbf{A} \in C$ cannot lie in the spherical cap $C_1$.
\end{lemma}
\begin{proof}
	By Prop. \ref{prop:conv_comb}, the center $c$ of $\mathcal{S}$ is a convex combination of at most $4$ points lying on the surface of $\mathcal{S}$.
	In other words, $c$ is the center of the circumscribed circle of either a line, a triangle or a tetrahedron.
	There are multiple cases that can occur regarding those points defining $\mathcal{S}$.

	First, consider the case that there is no subset of points \(\mathcal{P}' \subseteq \mathcal{P} \cup \{\mathbf{A}\}\) on the surface of $\mathcal{S}$ containing $\mathbf{A}$, such that $c$ is a convex combination of points from $\mathcal{P}'$.
	Then $c$ must lie outside of $C$, since it is a convex combination of points outside of $C$.

	Now consider the case that $\mathbf{A}$ is one of the points defining $c$, i.e. it is part of a set of points \(\mathcal{P}'\) such that $c$ is a convex combination of \(\mathcal{P}'\).
	If $\mathbf{A}$ does not lie in $C_1$, then $c$ once again cannot lie in $C_1$ either.
	Now assume that $\mathbf{A} \in C_1$ and consider multiple sub-cases depending on the cardinality of $\mathcal{P}'$:
	The first two cases are completely analogous to the proof of the two-dimensional version of the Lemma shown by Degener et al. \cite{DBLP:conf/spaa/DegenerKLHPW11}.
	\begin{itemize}
		\item Case $|\mathcal{P}'| = 2$: $\mathcal{S}$ is defined by $\mathbf{A}$ and another point $P \in \mathcal{P}$. Since $\mathcal{S}$ is a \ac{ses}, $c$ is the midpoint between $\mathbf{A}$ and $P$.
		Between $\mathbf{A}$ and $P$ lies the spherical segment $C_2$, with the same height as $C_1$, in which $\mathbf{A}$ lies.
		Therefore, the midpoint between them cannot lie in $C_1$.
		\item Case $|\mathcal{P}'| = 3$: $\mathcal{S}$ is the circumscribed circle of a triangle $\Delta \mathbf{A} P_1 P_2$ formed by $\mathbf{A}$ and two other points $P_1$ and $P_2$.
		The center of $\mathcal{S}$ can now be computed by determining the intersection of the perpendicular bisectors of two of the triangle's sides, say $\overline{\mathbf{A}P_1}$ and $\overline{\mathbf{A}P_2}$.
		Since $\mathbf{A}$ lies inside $C_1$ and the other two points lie outside of $C$, the centers of those two sides cannot lie in $C_1$.
		Furthermore, $c$ must lie inside of the triangle as it is a convex combination of its corners, meaning it is acute.
		Therefore, the intersection point of the perpendicular bisectors cannot lie in $C_1$.
		\item Case $|P'| = 4$: $\mathcal{S}$ is the circumscribed sphere of a tetrahedron $T$ formed by the points $\mathbf{A}$, $P_1$, $P_2$ and $P_3$.
		Similarly to the previous case, $c$ lies at the intersection of the perpendicular bisector planes of three edges of $T$.
		Choose those three edges as those adjacent to $\mathbf{A}$.
		Their midpoints once again must lie outside of $C_1$, since the points $P_1$, $P_2$ and $P_3$ lie outside of $C$.
		Since $c$ has to lie inside of $T$, and the three edges are all adjacent to the same point $\mathbf{A}$, their bisector planes can only meet below their midpoints, outside of $C_1$.
	\end{itemize}
	\qed
\end{proof}

\subsubsection{Hindering Robots}
In the previous subsection, several properties concerning the target points of robots have been established.
However, due to the second part of the strategy, robots are not always able to reach their target points.
In particular, if a robot $r_i$ has a neighbor $r_j$ such that $r_i$'s target point lies outside of the \emph{limit sphere} $\mathcal{B}_j(t)$ of $r_j$, then $r_i$ will be stopped by that limit sphere.
If this is the case, we say that $r_j$ \emph{hinders} $r_i$ from reaching its target point.

Next, it will be established that robots always reach their target points if they lie within $C$.

\begin{lemma}
	\label{Lemma:not_hindered_in_C}
	Robots that compute a target point in $C$ cannot be hindered from reaching it by any other robot.
\end{lemma}
\begin{proof}
	Let $r_i$ be a robot that computes a target point $c$ within $C$, which is the center of the \ac{ses} $\mathcal{S}_i$.
	Assume there is a robot $r_j$ in the neighborhood of $r_i$ that hinders $r_i$ from reaching $c$.

	By Corollary \ref{Cor:C_center_radius}, we know that $\mathcal{S}_i$ has a radius of at most $1/4$.
	Since $r_j$'s position must lie within $\mathcal{S}_i$, it can have a distance of at most $1/2$ to $r_i$.
	Therefore, the midpoint $m_j$ between $r_i$ and $r_j$'s positions is within a distance of at most $1/4$ from $r_i$.
	This leaves a distance of at least $1/2 - 1/4 = 1/4$ that $r_i$ can move freely without leaving the limit sphere of $r_j$.
	As $\mathcal{S}_i$ has a radius of at most $1/4$, this is enough to reach $c$, meaning $r_j$ cannot have hindered $r_i$.
	\qed
\end{proof}

Similarly, robots cannot be hindered from leaving $C$, as is shown by the following lemma.

\begin{lemma}
	\label{Lemma:not_hindered_leaving}
	Robots cannot be hindered from leaving $C$ by any other robot.
\end{lemma}
\begin{proof}
	Let $r_i$ be a robot in $C$ that computes a target point $c$ outside of $C$ and let $m$ be the point where $r_i$ would leave $C$.
	Note that $d(p_i,m) \leq 1/4$, since the diameter of $C$ is at most $1/4$.

	Now assume that there is a robot $r_j$ hindering $r_i$ from leaving $C$.
	First of all, $r_j$ must be a neighbor of $r_i$, meaning it is within distance $1$ of $r_i$ as well as $c$.
	Now, let $m_j = \frac{p_i + p_j}{2}$ be the midpoint between $r_i$ and $r_j$.
	For $r_j$ to hinder $r_i$ from reaching $m$, $m$ must lie outside of $r_j$'s limit sphere, i.e. $d(m, m_j) > 1/2$.
	Furthermore, let $p_i' = 2m - p_i$ be the reflection of $p_i$ about $m$.
	Note that $m = \frac{p_i' + p_i}{2}$.
	Putting these together yields:
	\begin{align*}
	&d(m,m_j) = ||m - m_j|| &> \frac{1}{2} \\
	\Longleftrightarrow\ & ||\frac{p_i'}{2} + \frac{p_i}{2} - \frac{p_i}{2} - \frac{p_j}{2}|| &> \frac{1}{2} \\
	\Longleftrightarrow\ & ||p_i' - p_j|| &> 1 \\
	\Longleftrightarrow\ & d(p_j, p_i') &> 1
	\end{align*}
	To summarize, we now have the following three constraints on the position of $p_j$:
	\begin{align}
	d(p_j, p_i) &\leq 1 \label{eq:pj_limit1}\\
	d(p_j, c) &\leq 1 \label{eq:pj_limit2}\\
	d(p_j, p_i') &> 1 \label{eq:pj_limit3}
	\end{align}
	From this we can conclude that $c$ must lie in between $p_i$ and $p_i'$ as follows:
	Assume $p_i'$ lies between $c$ and $p_i$ instead.
	By the constraints (\ref{eq:pj_limit1}) and (\ref{eq:pj_limit2}), $p_j$ must lie within the overlap of two balls of radius $1$ centered around $p_i$ and $c$.
	However, it must lie \emph{outside} of the ball of radius $1$ centered around $p_i'$ (\ref{eq:pj_limit3}).
	Since $p_i'$ lies on the line segment between $p_i$ and $c$ by assumption, the entire overlap of the first two balls must be contained within the last ball around $p_i'$, leaving no viable positions for $p_j$.
	Therefore, $p_i'$ cannot lie between $p_i$ and $c$ and thus, $c$ must lie between $p_i$ and $p_i'$ instead.

	\begin{figure}
		\label{fig:C3}
		\centering
		\includegraphics[width = \textwidth]{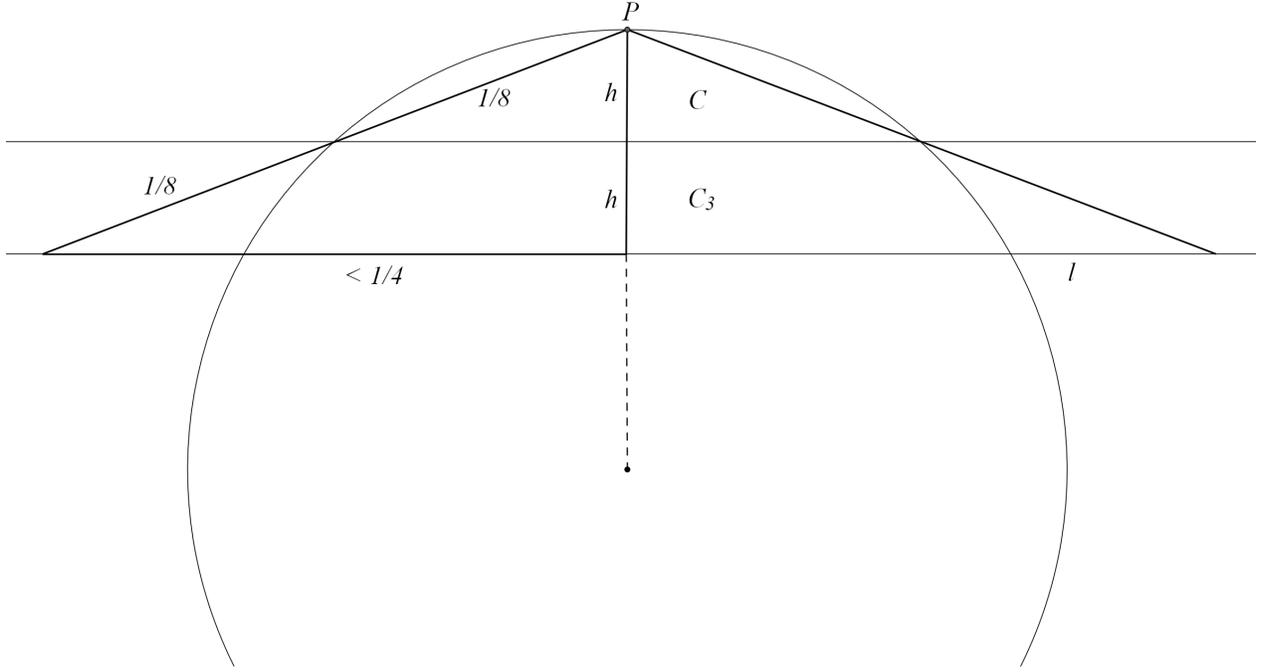}
		\caption{A cross-section of $\mathcal{N}$ illustrating the construction of the spherical segment $C_3$. It lies below and has the same height as $C$ and is defined by the base of $C$ and the plane $l$. Its diameter is at most $1/2$.}
	\end{figure}

	Based on what is now known about the position of $c$, another spherical segment can be defined that has to contain $p_i'$ and $c$:
	Let $C_3$ be the spherical segment of $\mathcal{N}$ below $C$ (see Fig. \ref{fig:C3}).
	It is defined by the base of $C$ and another plane $l$ which is parallel to and in distance $h$ of the base of $C$.
	The diameter of $C_3$ is now the diameter of the intersection of $l$ and $\mathcal{N}$ and can be bounded as follows:
	Once more consider the inscribed cone of $C$.
	Now extend this cone until it intersects with $l$.
	Note that this extended cone now has height $2h$ and slant height $1/4$ (since the slant height of the initial inscribed cone was $1/8$ by construction).
	Therefore, the radius of its base is also at most $1/4$.
	Since the cone intersects $\mathcal{N}$ at the base of $C$, it must intersect $l$ outside of $\mathcal{N}$, meaning the base of $C_3$ is completely contained within the base of the cone.
	Therefore, the diameter of $C_3$ can be at most $1/2$.

	Note that since $p_i'$ was defined as the reflection of $p_i$ about a point on the base of $C$, and $C_3$ has the same height as $C$, $p_i'$ and therefore also $c$ must be contained in $C_3$.
	Thus, we have $d(c, p_i') \leq 1/2$.
	Using Lemma \ref{Lemma:SES_cap_bound} and the diameter of $C_3$, the radius of the \ac{ses} corresponding to the target point $c$ can now be further bounded to be at most $1/2$, i.e. $d(p_j, c) \leq 1/2$.
	Putting these two bounds together and applying the triangle inequality yields:

	\begin{align*}
	d(p_j, p_i') \leq d(p_j,c) + d(c, p_i') \leq \frac{1}{2} + \frac{1}{2} = 1
	\end{align*}

	This is a contradiction to $d(p_j, p_i') > 1$ (\ref{eq:pj_limit3}), meaning $r_j$ cannot have hindered $r_i$ from leaving $C$.
	\qed
\end{proof}

\subsubsection{The Upper Bound}
Now that the necessary preliminaries have been established, putting them together to acquire the main result works completely analogously to the proof for the two-dimensional case shown by Degener et al.\ \cite{DBLP:conf/spaa/DegenerKLHPW11}.

\begin{lemma}
	\label{Lemma:pairwise_diff_robots}
	Let $\mathcal{P}$ be a set of robots in round $t_0$ that are all positioned in or compute a target point in $C$ and that all have pairwise different neighborhoods.
	Then at most one of those robots is in $C$ at the beginning of the next round.
\end{lemma}
\begin{proof}
	Let $r_i$ be a robot whose set of neighbors $\mathcal{R}_i(t_0)$ is minimal, i.e. no other robot $r_j \in \mathcal{P} \setminus \{r_i\}$ has a neighborhood $\mathcal{R}_j(t_0)$ that is a strict subset of $\mathcal{R}_i(t_0)$.
	Therefore, all other robots $r_j$ have a neighbor $r_k$ that is not seen by $r_i$, i.e. $d(r_i, r_k) > 1$.

	First consider the case that $r_i$ is positioned in $C$:
	In this case Lemma \ref{Lemma:dist_point} can be applied to show that no neighbor $r_j$ of $r_i$ can compute a target point in $C$ by choosing $\mathcal{R}_j(t_0)$ as the point set $\mathcal{P}$, the position of $r_i$ as the point $\mathbf{A}$ and the position of $r_k$ as the point $\mathbf{B}$ in distance greater than $1/2$ of $\mathbf{A}$.
	Note that any robot outside of $C$ that could potentially compute a target point in $C$ must be within distance $1/4$ of $C$ by Corollary \ref{Cor:C_center_radius} and is therefore at most a distance of $1/2$ away from $r_i$ and thus its neighbor.
	It follows that only $r_i$ could possibly remain in $C$.

	On the other hand, if $r_i$ is positioned outside of $C$, $r_i$ computes a target point in $C$ by the definition of $\mathcal{P}$.
	By Corollary \ref{Cor:C_center_radius}, the \ac{ses} of $\mathcal{R}_i(t_0)$ now has a radius of at most $1/4$, meaning $r_i$ is also at most a distance of $1/4$ away from $C$ and also in distance at most $1/2$ from any point in $C$.
	Now again consider a robot $r_j \in \mathcal{P} \setminus \{r_i\}$ with a neighbor $r_k$ that is unseen by $r_i$.
	For this robot $r_k$, we must have $d(r_k, C) > 1/2$, otherwise it would be a neighbor of $r_i$, since it would be in distance $1$ of it by applying the triangle inequality:
	\begin{align*}
	d(r_k, C) + d( r_i, C) \leq 1 < d(r_i, r_k)
	\end{align*}
	From this, it can be concluded that $r_j$'s target point cannot lie in $C$:
	If it was in $C$, the radius of $\mathcal{R}_j(t_0)$'s \ac{ses} would also be at most $1/4$ by Corollary \ref{Cor:C_center_radius}, meaning it could not contain $r_k$.
	This means that $r_j$ is positioned in $C$ by definition of $\mathcal{P}$ and computes a target point outside of $C$.
	By Lemma \ref{Lemma:not_hindered_leaving}, $r_j$ cannot be hindered from leaving $C$.
	Since this holds for all robots except for $r_i$, $r_i$ is once again the only robot that might remain in $C$ in the following round $t_0 + 1$.
	\qed
\end{proof}

Using this lemma, it is now finally possible to show that the progress measures mentioned earlier in this section always improve during a pair of consecutive rounds:

\begin{lemma}
	\label{Lemma:progress}
	If $R_t \geq 1/2$, either there are robots that merge in round $t$ or after rounds $t$ and $t+1$, the cap $C_1$ does not contain any robots.
\end{lemma}
\begin{proof}
	Consider all robots that are positioned in or compute a target point in $C$.
	We can distinguish two types of robots:
	Firstly, there are robots $r_i$ that have a neighbor $r_j$ with the same neighborhood, i.e. $\mathcal{R}_i(t) = \mathcal{R}_j(t)$.
	This means that $r_i$ and $r_j$ compute the same target point.
	If there is such a pair of robots that compute a common target point in $C$, then by Lemma \ref{Lemma:not_hindered_in_C}, they both reach it and thus merge and fulfill the Lemma to be proven.
	On the other hand, if such a pair does not exist, then all robots sharing neighborhoods with other robots must compute target points outside of $C$ and are not hindered from leaving $C$ (Lemma \ref{Lemma:not_hindered_leaving}) if they were positioned inside of it, meaning none of them can remain within $C$ at the end of the round.

	It remains to consider the set of robots that all have pairwise different neighborhoods and also either lie in or compute a target point in $C$.
	This is the exact situation described by Lemma \ref{Lemma:pairwise_diff_robots}, meaning at most one of these robots can remain within $C$ after round $t$.

	Therefore, in the beginning of round $t + 1$ and if no robots merged in round $t$, at most one robot $r_i$ remains within $C$.
	If it lies in $C_2$, we are done.
	Otherwise, if it lies in $C_1$, only $r_i$ itself and its neighbor could possibly compute a target point in $C_1$.
	However, by Lemma \ref{Lemma:single_C1_point}, this cannot happen and by Lemmas \ref{Lemma:not_hindered_in_C} and \ref{Lemma:not_hindered_leaving}, $r_i$ cannot be hindered from leaving $C_1$ at which point no robots remain in $C_1$.
	\qed
\end{proof}

Using this lemma now yields the main result.

\mainTheorem*
\begin{proof}
	Fix an arbitrary round $t_0 \geq 0$.
	Lemma \ref{Lemma:progress} holds for any spherical cap $C$ with an arbitrary point $P$ on the boundary of the global \ac{ses} $\mathcal{N}_{t_0}$ as its apex.
	Therefore, either at least two robots merge or all robots robots are within distance $h/2$ of the boundary of $\mathcal{N}_{t_0}$ at the beginning of round $t_0 + 2$.
	This means that the radius of $\mathcal{N}_{t_0 + 2}$ is at least $h/2$ smaller than that of $\mathcal{N}_{t_0}$ since it is a \ac{ses}.
	By Lemma \ref{Lemma:height_bound}, we have:
	\begin{align*}
	\frac{h}{2} \geq \frac{1}{128 \pi \cdot R_{t_0}} \geq \frac{1}{128 \pi \cdot R_0} \geq \frac{1}{128 \pi \cdot n}
	\end{align*}
	Consequently, it takes at most $\lceil 2 \cdot 128 \cdot \pi \cdot n^2 \rceil$ rounds without merging robots until the radius is less than $1/2$, at which point all robots can see each other, compute the same target and move towards it in a single round.

	Therefore, there can overall be either at most $n-1$ rounds with merges or $\mathcal{O}(n^2)$ rounds without merges until all robots have gathered.
	\qed
\end{proof}




%
%

\end{document}